\documentclass{article}

\usepackage[english]{babel}
\usepackage{color}

\usepackage[letterpaper,top=2cm,bottom=2cm,left=3cm,right=3cm,marginparwidth=1.75cm]{geometry}

\newcommand{\R}{{\mathbb R}}  %ams bold
\newcommand{\Pro}{{\mathbb P}}  %ams bold

\usepackage{graphicx}
\usepackage[colorlinks=true, allcolors=blue]{hyperref}
\usepackage{amsmath}
\usepackage{amsfonts}
\usepackage{amsthm}
\usepackage{amssymb}
\usepackage{dsfont}
\usepackage{subcaption}
\newtheorem{assumption}{Assumption}[section]

\newtheorem{theorem}{Theorem}[section]

\newtheorem{lemma}{Lemma}[section]

\newtheorem{remark}{Remark}[section]
\newtheorem{definition}{Definition}[section]

\newcommand\norm[1]{\lVert#1\rVert}
\newcommand\normF[1]{\lVert#1\rVert_{\mathrm{F}}}
\newcommand\normSup[1]{\lVert#1\rVert_{\infty}}
\newcommand\Tr[1]{\mathrm{Tr}\left(#1\right)}

\newcommand\eigmax[1]{\lambda_{\mathrm{max}}\left(#1\right)}
\newcommand\eigmin[1]{\lambda_{\mathrm{min}}\left(#1\right)}
\newcommand{\de}{\mathrm{d}}
\newcommand\innprod[2]{\langle #1, #2 \rangle}
\DeclareMathOperator*{\esssup}{ess\,sup}
\DeclareMathOperator*{\sym}{\mathbb{S}}
\newcommand\E{\mathbb{E}}
\newcommand\Prob{\mathbb{P}}
\newcommand\Vcal{\mathcal{V}}
\newcommand\Lcal{\mathcal{L}}
\newcommand\Scal{\mathcal{S}}
\newcommand\Jcal{\mathcal{J}}
\newcommand\Zcal{\mathcal{Z}}
\newcommand\Kcal{\mathcal{K}}
\newcommand\Acal{\mathcal{A}}

\newcommand\Dcal{\mathcal{D}}
\newcommand\Gcal{\mathcal{G}}

\newcommand\PDcal{\mathcal{PD}}

\newcommand\Idc{\mathbf{1}}

\title{Small-Covariance Noise-to-State Stability of Stochastic Systems and Its Applications to Stochastic Gradient Dynamics  %
\thanks{This work is supported by the University of New Mexico under the School of Engineering (SOE) faculty startup for L.C., the National Science Foundation (NSF) grants CNS-2227153 and ECCS-2210320 for Z.P.J., and the Air Force Office of Scientific Research (AFOSR) grant FA9550-21-1-0289 and the Office of Naval Research (ONR) grant N00014-21-1-2431 for E.D.S.}
}
\author{Leilei Cui\footnote{Department of Mechanical Engineering, University of New Mexico, Albuquerque, NM, USA},\, Zhong-Ping Jiang\footnote{Department of Electrical and Computer Engineering and Department of Civil and Urban Engineering, New York University, Brooklyn, NY, USA},\, and Eduardo D. Sontag\footnote{Department of Electrical and Computer Engineering and Department of BioEngineering, Northeastern University, Boston, MA, USA}}

\begin{document}
\maketitle

\begin{abstract}
This paper studies gradient dynamics subject to additive random noise, which may arise from sources such as stochastic gradient estimation, measurement noise, or stochastic sampling errors. To analyze the robustness of such stochastic gradient systems, the concept of small-covariance noise-to-state stability (NSS) is introduced, along with a Lyapunov-based characterization. Furthermore, the classical Polyak–Łojasiewicz (PL) condition on the objective function is generalized to the $\Kcal$-PL condition via comparison functions, thereby extending its applicability to a broader class of optimization problems. It is shown that the stochastic gradient dynamics exhibit small-covariance NSS if the objective function satisfies the $\Kcal$-PL condition and possesses a globally Lipschitz continuous gradient. This result implies that the trajectories of stochastic gradient dynamics converge to a neighborhood of the optimum with high probability, with the size of the neighborhood determined by the noise covariance. Moreover, if the $\Kcal$-PL condition is strengthened to a $\Kcal_\infty$-PL condition, the dynamics are NSS; whereas if it is weakened to a general positive-definite-PL condition, the dynamics exhibit integral NSS. The results further extend to objectives without globally Lipschitz gradients through appropriate step-size tuning. The proposed framework is further applied to the robustness analysis of policy optimization for the linear quadratic regulator (LQR) and logistic regression.

\end{abstract}

\section{Introduction}
Optimization lies at the core of many data-driven fields, providing concepts and tools for algorithm design, computational-complexity analysis, and statistical inference. Recent work in learning and optimization largely focuses on gradient-based methods because of their low per-iteration cost and suitability for parallel architectures. Their scalability to large-scale problems has, in turn, motivated careful study of convergence rates—both how to certify a target rate and how to systematically improve it through choices such as step-size policies and momentum. Yet efficiency alone is insufficient: in practice, gradient-based methods operate in noisy environments, where perturbations arise from numerical errors, measurement noise, inexact formulas for gradient computations, and early stopping of embedded routines used to compute gradients (see \cite[Ch.~4]{book_Polyak} and \cite[p.~38]{bertsekas1997nonlinear}). Under such unpredictable noise, the gradient-based algorithms may oscillate around the optimum, converge to a biased limit point, or even diverge. Hence, convergence and robustness are twin, indispensable facets that must be analyzed for gradient-based methods. 

The connection between optimization and control enables the study of algorithmic convergence through stability and Lyapunov theory \cite{polyak2017lyapunov}. By modeling a gradient method as a continuous-time dynamical system and treating the iterates as system states, one can employ Lyapunov functions to characterize asymptotic behavior. Importantly, flexibility in Lyapunov design—beyond using only the objective value or distance to an optimizer/optimal set—enriches the analytical toolkit and often yields sharper bounds and insights than any single canonical choice \cite{WibisonoWilsonJordan2016,wilson2016lyapunov,Jordan2018DynamicalSymplecticStochastic}. Control-theoretic insights can also guide the redesign of optimization algorithms themselves \cite{chen2024accelerated}. Classic Lyapunov methods are typically formulated for systems without exogenous inputs/disturbances. Input-to-state stability (ISS) generalizes Lyapunov stability to systems with inputs, explicitly quantifying how disturbances impact the state \cite{Sontag1989}; it thus provides a natural framework for analyzing both convergence and robustness of gradient-based optimization under external perturbations. For example, \cite{ieee_tac_2018_cherukuri_et_al_convexity_saddle_point_dynamics} establishes ISS of the saddle-point dynamics for convex–concave objectives with respect to additive noise; \cite{kolmanovsky2022inputtostate} uses ISS to assess the robustness of a bilevel optimization algorithm to inner-loop approximation errors; ISS likewise underpins robustness analyses of extremum-seeking methods \cite{2021_poveda_krstic_fixedtime_iss_extremum_seeking,2021_arxiv_iss_gradient_suttner_dashkovskiy}; and \cite{2020arxiv_bianchin_poveda_dallanese_gradient_iss_switched_systems} leverages ISS to address output regulation for tracking a gradient flow in the presence of plant-level disturbances. Finally, \cite{CJS2024,cui2025perturbed} generalizes the classical Polyak–\L{}ojasiewicz (PL) condition \cite{Polyak1963,Lojasiewicz1963,Hamed2016} and proves that gradient descent is small-disturbance ISS when the objective satisfies a generalized $\Kcal$-PL condition. Most of the above focuses on deterministic disturbances. In practice, however, noise is often random: objective values and gradients are obtained via measurements subject to stochastic errors, and in stochastic optimization the objective is an expectation, so exact gradients are typically unavailable and must be approximated by random samples. This motivates studying the impact of randomness on gradient-based methods. Noise-to-state stability (NSS) extends ISS to stochastic systems, offering a meaningful framework when bounding the state by the supremum of unbounded white noise is not feasible \cite{Deng-Krstic-Williams2001}. NSS asserts that trajectories enter (and remain in) a noise-dependent neighborhood of the equilibrium with high probability, with the radius controlled by the noise intensity/covariance. In this paper, we study the robustness of gradient dynamics under additive stochastic perturbations within the NSS framework.

The \emph{Polyak–\L{}ojasiewicz (PL)} condition \cite{Polyak1963,Lojasiewicz1963,Hamed2016} is a useful and straightforward tool for certifying a linear convergence rate of gradient-based optimization without the assumption of strong convexity. However, this condition requires that the norm of the gradient of the objective function is larger than the square root of its suboptimality, i.e., $\norm{\nabla \Jcal} \ge \sqrt{c(\Jcal(z) - \Jcal^*}$ with $c > 0$. However, this square-root scaling between gradient norm and suboptimality can be restrictive and is violated in many important problems (e.g., policy optimization for the linear–quadratic regulator and logistic regression), which we analyze in the second part of this paper. In \cite{CJS2024,cui2025perturbed}, the authors generalized the classic PL condition to a $\Kcal$-PL condition by requiring that $\norm{\nabla \Jcal} \ge \mu(\Jcal(z) - \Jcal^*)$, where $\mu$ is a $\Kcal$-function (continuous, strictly increasing, and vanishing at zero). If such a $\Kcal$-PL condition is satisfied, the gradient dynamics is small-disturbance ISS; however, the analysis there is limited to deterministic noise. In this paper, we move a step further and show that, if the objective function satisfies the $\Kcal$-PL condition, the stochastic gradient dynamics (overdamped Langevin diffusion) with time-varying covariance is small-covariance NSS. This implies that, if the noise covariance is sufficiently small, the trajectories of the stochastic gradient dynamics will eventually enter and remain in a neighborhood of the optimum with high probability, with the neighborhood size depending (nonlinearly) on the covariance. In addition, if the $\Kcal$-PL condition is strengthened to a $\Kcal_\infty$-PL condition—requiring $\mu$ to be unbounded—then the stochastic dynamics is NSS. If the  $\Kcal$-PL condition is weakened to a positive-definite function (removing the monotonicity requirement on $\mu$), then the stochastic dynamics is integral NSS. Similar results hold for the underdamped Langevin diffusion, which is the continuous-time version of the heavy-ball optimization algorithm \cite{Polyak1964HeavyBall}, subject to stochastic noise with time-varying covariance.

The results are applied first to policy optimization for LQR. In LQR, the goal is to optimize the feedback control gain to minimize the cumulative cost of quadratic state and input terms. This problem provides an ideal benchmark for theoretically analyzing the performance of policy-gradient algorithms in reinforcement learning, where the control policy is updated via gradient descent to reduce cost (or gradient ascent to increase reward) \cite[Ch. 13]{book_sutton}. In recent studies \cite{fazel2018global,Mohammadi2022,Hu2023Review}, the convergence of PO to optimality has been well studied, but the robustness of PO under stochastic noise remains an open problem that should be systematically investigated. In this paper, by resorting to the established 
$\Kcal$-PL condition of LQR cost \cite{CJS2024}, we show that the stochastic gradient dynamics for PO of LQR is small-covariance NSS (or NSS) if the step size/learning rate is chosen appropriately. The learning rate must be selected carefully because the gradient of the objective function is not globally Lipschitz continuous. The second application of the developed robustness-analysis framework is to logistic regression. Under mild conditions (the data points are nonseparable and full rank), the logistic loss is coercive and has a globally Lipschitz-continuous gradient. In addition, even though the logistic loss is strictly convex, there is no function that satisfies the classic PL condition nor the 
$\Kcal_\infty$-PL condition. We further show that the logistic loss can satisfy the weaker 
$\Kcal$-PL condition, and hence the corresponding stochastic gradient dynamics (including both overdamped and underdamped Langevin diffusion) are small-covariance NSS. 

The remainder of the paper is organized as follows. In Section 2, we introduce notation and comparison functions. In Section 3, we present the concepts of NSS—particularly the new notion of small-covariance NSS—and provide a Lyapunov sufficient condition to ensure small-covariance NSS. In Section 4, we analyze the robustness of stochastic systems useful in optimization (referred to as overdamped and underdamped Langevin diffusions) and establish connections between the various NSS notions and generalized PL conditions. In Section 5, we apply the main results to policy optimization for LQR and to logistic regression. Section 6 concludes the paper.

\section{Notations and Facts}
\subsection{Notations}
In this article, let $\R$ and $\R_+$ denote the set of real numbers and the set of nonnegative real numbers, respectively. Let $\R^n$ denote the $n$-dimensional Euclidean space. The sets $\sym^n$, $\sym_+^n$ and $\sym^n_{++}$ denote the space of $n \times n$ real symmetric matrices, real symmetric positive semi-definite matrices, and real symmetric positive definite matrices, respectively. The sets of integers and positive integers are denoted by $\mathbb{Z}$ and $\mathbb{Z}_+$, respectively. The trace of a square matrix is denoted by $\Tr{\cdot}$. The symbol $\norm{\cdot}$ denotes the Euclidean norm of a vector or the spectral norm of a matrix, while $\normF{\cdot}$ denotes the Frobenius norm of a matrix, where $\normF{A}^2 = \Tr{A^\top A}$ for $A \in \R^{n \times m}$. For a measurable and essentially bounded function $x: \R_+ \to \R^n$ (or $X: \R_+ \to \R^{n \times m}$), its essential supremum norm is denoted by $\norm{x}_\infty = \esssup_{t \in \R_+} \norm{x(t)}$ (or $\norm{X}_\infty = \esssup_{t \in \R_+} \norm{X(t)}$). 

Let $\mathcal{C}(\mathcal{S}, \R^n)$, $\mathcal{C}^1(\mathcal{S}, \R^n)$ and $\mathcal{C}^2(\mathcal{S}, \R^n)$ denote the spaces of continuous, continuously differentiable, and twice continuously differentiable functions from $\mathcal{S} \subset \R^n$ to $\R^n$, respectively. For a function $\mathcal{V} \in \mathcal{C}^2(\mathcal{S}, \R^n)$, its gradient and Hessian are denoted by $\nabla \mathcal{V}$ and $\nabla^2 \mathcal{V}$, respectively. The symbols $\eigmin{\cdot}$ and $\eigmax{\cdot}$ denote the minimum and maximum eigenvalues of a real symmetric matrix, respectively. The $n$-dimensional identity matrix is denoted by $I_n$ and $\mathrm{Id}$ denotes the identity function. For any $X_1, X_2 \in \R^{m \times n}$ and $Y \in \sym^n_{++}$, define the inner product as $\innprod{K_1}{K_2}_Y = \Tr{K_1 Y K_2^\top}$. For simplicity, we write $\innprod{K_1}{K_2}_{I_n} = \innprod{K_1}{K_2}$. For any $A,B \in \sym^n$, the notion $A \succ B$ ($A \succeq B$) means that $A - B \in \sym^n_{++}$ ($A - B \in \sym^n_{+}$). Similarly, $A \prec B$ and $A \preceq B$ indicate that $B - A \in \sym^n_{++}$ and $B - A \in \sym^n_{+}$, respectively. Let $(\Omega, \mathcal{F}, \{\mathcal{F}_t\}_{t\ge 0}, \mathbb{P})$ be a complete filtered probability space, where $\Omega$ is the sample space, $\mathcal{F}$ is a $\sigma$-algebra over $\Omega$, $\{\mathcal{F}_t\}_{t\ge 0}$ is a filtration satisfying the usual conditions (right-continuity and completeness), and $\mathbb{P}$ is a probability measure defined on $(\Omega, \mathcal{F})$. $\mathbf{1}_A(\omega)$ denotes the indicator function of set $A \subset \Omega$.

\subsection{Notions of Comparison Functions} 
The notions of comparison functions are introduced to facilitate the stability analysis of dynamical systems. A function $\alpha: \R_+ \to \R_+$ is said to be positive definite ($\mathcal{PD}$) if $\alpha(0)=0$, and $\alpha(r)>0$ for all $r \neq 0$. A function $\alpha: \R_+ \to \R_+$ is said to be of class $\mathcal{K}$ if it is continuous, strictly increasing, and satisfies $\alpha(0)=0$. It is said to be of class $\mathcal{K}_\infty$ if $\alpha \in \mathcal{K}$ and $\alpha(r) \to \infty$ as $r\to \infty$. A function $\alpha$ is said to be of class $\Kcal_{[0,d)}$ for some $d>0$ if it is defined on $[0,d)$, continuous, strictly increasing, and satisfies $\alpha(0)=0$. A function $\beta: \R_+ \times \R_+ \to \R_+$ is said to be of class $\mathcal{KL}$ if for each $t \ge 0$, the function $\beta(\cdot, t) \in \mathcal{K}$, and for each fixed $r \ge 0$, $\beta(r,\cdot)$ is decreasing and satisfies $\lim_{t\to \infty}\beta(r,t) = 0$.

\subsection{Useful Facts}

\begin{lemma}[Weak triangle inequality in \cite{Jiang1994}]\label{lm:weakTriangle}
    For any $\mathcal{K}$-function $\alpha$, any $\mathcal{K}_\infty$-function $\rho$, and any $a,b \in \R_+$, it holds that
    \begin{align*}
        \alpha(a+b) \le \alpha \circ (\mathrm{Id} + \rho)(a) + \alpha \circ (\mathrm{Id} + \rho^{-1}) (b).
    \end{align*}
\end{lemma}

\begin{lemma}[Trace inequality \cite{Wang1986}]\label{lm:traceIneq}
    For any $S\in \mathbb{S}^{n}$ and $P \in  \mathbb{S}^{n}_+$, it holds
    \begin{align*}
        \eigmin{S} \Tr{P} \le \Tr{SP} \le \eigmax{S} \Tr{P}.
    \end{align*}
\end{lemma}

\section{Preliminaries of Noise-to-State Stability}
\subsection{Introduction of Nonlinear Stochastic Systems}

Consider the following nonlinear stochastic system defined on the probability space $(\Omega, \mathcal{F}, \{\mathcal{F}_t\}_{t \ge 0}, \mathbb{P})$
\begin{align}\label{eq:nonLinearSys}
    \de \chi(\omega, t) = f(\chi(\omega, t)) \de t + g(\chi(\omega, t)) \Sigma(t) \de B(\omega ,t), 
\end{align}
where $\chi: \Omega \times [0, t_{max}) \to \mathcal{S} \subset \R^n$ is an $n$-dimensional stochastic process; the state space $\mathcal{S} \subset \R^n$ is an open subset diffeomorphic to $\R^n$; the process $B: \Omega \times [0,+\infty) \to \R^{m}$ denotes a standard $m$-dimensional Brownian motion; the drift field $f: \mathcal{S} \subset \R^n \to \R^n$ and the diffusion field $g: \mathcal{S} \subset \R^n \to \R^{n \times m}$ are both locally bounded and locally Lipschitz continuous; $f(\chi^*)=0$ for some equilibrium  $\chi^* \in \R^n$, and the matrix-valued function  $\Sigma: \R_+ \to \R^{m \times m}$ is Borel measurable and locally essentially bounded. At each time $t \in \R_+$, $\Sigma(t)$ linearly transforms the Brownian noise $\de B(\omega, t)$, resulting in an input $\Sigma(t)\de B(\omega, t)$ with instantaneous covariance $\Sigma(t)\Sigma(t)^\top \de t$. In the following, we refer to $\norm{\Sigma(t)\Sigma(t)^\top}$ as the instantaneous noise intensity. The It\^{o} integral form of \eqref{eq:nonLinearSys} is given by
\begin{align}\label{eq:inteNonlinear}
    \chi(\omega, t) = \int_{0}^{t}f(\chi(\omega, s)) \de s + \int_{0}^{t}g(\chi(\omega, s)) \Sigma(s) \de B(\omega ,s)
\end{align}
where the second integral term is interpreted as an It\^{o} stochastic integral. A stochastic process $\{ \chi(t)\}_{t \ge 0}$ is said to be a strong solution of \eqref{eq:nonLinearSys} with initial condition $\chi_0 \in \Scal$ if the following conditions are satisfied \cite[Definition 2.1]{Book_karatzas1991}:
\begin{enumerate}
    \item the sample paths of $\{ \chi(t)\}_{t \ge 0}$ are continuous, and the process is adapted to the filtration $\{\mathcal{F}_t\}_{t \ge 0}$; 
    \item $\mathbb{P}[\chi(0) = \chi_0]=1$;
    \item $\mathbb{P}[\int_{0}^t |f_i(\chi(s))| + |g_{i,j}(\chi(s))|\de s < \infty]=1$ for every $1 \le i \le n$, $1 \le j \le m$, and any $t \in \R_+$; 
    \item the integral equation \eqref{eq:inteNonlinear} holds almost surely for every $t \in \R_+$.
\end{enumerate}
According to \cite[p. 287]{Book_karatzas1991}, the local Lipschitz continuity of $f$ and $g$ ensures the local existence (up to an explosion time $t_{max}$) and uniqueness of a strong solution to \eqref{eq:nonLinearSys}. 

The infinitesimal generator is introduced below to facilitate the stability analysis of the nonlinear stochastic system.
\begin{definition}
    The infinitesimal generator $\mathcal{L}$ associated with system \eqref{eq:nonLinearSys}, acting on a function $\mathcal{V} \in \mathcal{C}^2(\mathcal{S}, \R_+)$, is defined as a mapping $\mathcal{L}[\mathcal{V}]: \mathcal{S} \times \R^{m \times m} \to \R$ given by
    \begin{align}
        \mathcal{L}[\mathcal{V}](\xi, \Theta) = \innprod{\nabla \mathcal{V}(\xi)}{f(\xi)} + \frac{1}{2}\innprod{g(\xi)\Theta }{\nabla^2 \mathcal{V}(\xi)g(\xi)\Theta}
    \end{align}
\end{definition}
The infinitesimal generator $\mathcal{L}[\mathcal{V}](\xi, \Theta)$ can be interpreted as the rate of change of the expected value of $\mathcal{V}$ at state $\xi$ under the influence of noise covariance $\Theta \Theta^\top$. It serves as a stochastic counterpart to the Lie derivative. By It\^{o}'s formula, along the stochastic process $\{\chi(t)\}_{t\ge 0}$ governed by \eqref{eq:nonLinearSys}, the evolution of $\mathcal{V}(\chi(t))$ is given by
\begin{align}\label{eq:ItoIntegral}
    \mathcal{V}(\chi(t)) = \mathcal{V}(\chi(0)) + \int_{0}^t \mathcal{L}[\mathcal{V}](\chi(s), \Sigma(s)) \de s + \int_{0}^t \innprod{\nabla \mathcal{V}(\chi(s))}{g(\chi(s)) \Sigma(s) \de B(s)}
\end{align}

\subsection{Notions of Noise-to-State Stability}
Since system \eqref{eq:nonLinearSys} is defined on an open subset rather than the entire Euclidean space, a size function is introduced to facilitate stability analysis within the open subset and to ensure that system trajectories remain inside the domain.

\begin{definition}\label{def:sizeFunc} (Size Function)
    A function $\mathcal{V}: \mathcal{S} \to \mathbb{R}_+$ is a size function for $(\mathcal{S},\chi^*)$ if $\mathcal{V}$ is 
    \begin{enumerate}
        \item twice continuously differentiable;
        \item positive definite with respect to $\chi^*$, i.e. $\mathcal{V}(\chi^*)=0$ and $\mathcal{V}(\chi)>0$ for all $\chi \neq \chi^*$,  $\chi \in \mathcal{S}$;
        \item coercive, i.e. for any sequence $\{\chi_k\}_{k=0}^\infty$, $\chi_k \to \partial \mathcal{S}$ or $\norm{\chi_k} \to \infty$, it holds that $\mathcal{V}(\chi_k) \to \infty$, as $k \to \infty$.
    \end{enumerate}
\end{definition}

For system \eqref{eq:nonLinearSys}, the presence of persistent additive stochastic noise—even at equilibrium—prevents trajectories from converging to the equilibrium point in the conventional asymptotic sense. Consequently, the standard notion of stochastic asymptotic stability—requiring that trajectories remain within a ball of radius $\gamma(\norm{\chi_0})$ for some $\gamma \in \Kcal$ with high probability and that $\lim_{t \to \infty}\norm{\chi(t)} = 0$ almost surely—no longer applies. Instead, one typically expects the state trajectories to eventually settle into a neighborhood of the equilibrium with high probability, where the size of the neighborhood depends on the noise intensity, quantified by magnitude of the noise covariance $\normSup{\Sigma \Sigma^\top}$. The concept of noise-to-state stability (NSS) provides a natural extension of input-to-state stability to stochastic systems, offering a meaningful framework in stochastic scenarios where bounding the state directly by the supremum of unbounded white noise is not feasible.

\begin{definition}(Noise-to-State Stability in Probability \cite{Deng-Krstic-Williams2001})
    System \eqref{eq:nonLinearSys} is NSS in probability if there exists a size function $\mathcal{V}$ such that for each $\epsilon \in (0,1)$, the following holds
    \begin{align}\label{eq:nss}
        \Pro\big\{ \mathcal{V}(\chi(t)) \le \beta(\mathcal{V}(\chi(0)), t) + \gamma(\normSup{\Sigma \Sigma^\top} )  \big\} \ge 1-\epsilon
    \end{align}
    for some $\beta \in \mathcal{KL}$, $\gamma \in \mathcal{K}$, all $x(0) \in \mathcal{S} $, and any $t \in \R_+$.
\end{definition}
By the causality of the dynamical system, the definition remains unchanged if $\normSup{\Sigma \Sigma^\top}$ is replaced by $\esssup_{\tau \in [0,t]} \norm{\Sigma(\tau)\Sigma(\tau)^\top}$. The aforementioned notion of NSS requires that trajectories remain bounded with high probability for arbitrarily large noise intensity (i.e. $\normSup{\Sigma \Sigma^\top}$). However, in practice, many stochastic systems may diverge under large noise covariance and only exhibit NSS behavior when the noise covariance is sufficiently small. This observation motivates the introduction of small-covariance NSS, which relaxes the classical NSS condition by characterizing stability properties that hold only under noise with small covariance.

\begin{definition}(Small-Covariance Noise-to-State Stability  in Probability )
    System \eqref{eq:nonLinearSys} is small-covariance noise-to-state stable (scNSS) in probability if there exists a size function $\mathcal{V}$ and a constant $d > 0$, such that for each $\epsilon \in (0,1)$, the following holds 
    \begin{align}
        \Pro\big\{ \mathcal{V}(\chi(t)) \le \beta(\mathcal{V}(\chi(0)), t) + \gamma(\normSup{\Sigma \Sigma^\top})  \big\} \ge 1-\epsilon
    \end{align}
    for some $\beta \in \mathcal{KL}$, $\gamma \in \mathcal{K}_{[0,d)}$, all noise covariance bounded by $d$ (i.e. $\norm{\Sigma \Sigma^\top}_\infty < d$), all $x(0) \in \mathcal{S}$, and any $t \in \R_+$.
\end{definition}

Analogous to integral input-to-state stability (iISS) in deterministic systems, the state trajectories can be bounded by an energy-like functional of the noise intensity. This observation motivates the notion of integral noise-to-state stability (iNSS) for stochastic systems.
\begin{definition} (Integral Noise-to-State Stability in Probability  \cite{Ito2020})
    System \eqref{eq:nonLinearSys} is integral noise-to-state stable (iNSS) in probability if there exists a size function $\mathcal{V}$ such that for each $\epsilon \in (0,1)$, the following holds
    \begin{align}
        \Pro\big\{ \mathcal{V}(\chi(t)) \le \beta(\mathcal{V}(\chi(0)), t) + \int_{0}^t \gamma(\norm{\Sigma(\tau)\Sigma(\tau)^\top}) \de \tau   \big\} \ge 1-\epsilon
    \end{align}
    for some $\beta \in \mathcal{KL}$, $\gamma \in \mathcal{K}$, all $x(0) \in \mathcal{S}$, and any $t \in \R_+$.
\end{definition}

\subsection{Lyapunov Functions of Noise-to-State Stability}
Lyapunov functions offer a tractable approach to determine whether a stochastic system satisfies NSS. The notion of Lyapunov function for NSS is introduced below.

\begin{definition}\label{def: nss_lf} (NSS-Lyapunov Function)
    A function $\mathcal{V}: \mathcal{S} \to \R_+$ is an NSS-Lyapunov function if 
    \begin{enumerate}
        \item $\mathcal{V}: \mathcal{S} \to \R_+$ is a size function for $(\mathcal{S}, \chi^*)$;
        \item there exist $\alpha \in \mathcal{K}_\infty$  and $\gamma \in \mathcal{K}$ such that
        \begin{align}\label{eq: lyapnovFun}
            \mathcal{L}[\mathcal{V}](\xi,\Theta) \le -\alpha(\mathcal{V}(\xi)) + \gamma(\norm{\Theta \Theta^\top})
        \end{align}
        for all $\xi \in \mathcal{S}$ and all $\Theta \in \R^{m \times m}$.
    \end{enumerate} 
\end{definition}

A scNSS-Lyapunov function can be derived from an NSS-Lyapunov function by relaxing the $\mathcal{K}_\infty$-function $\alpha$ in Definition \ref{def: nss_lf} to a general $\mathcal{K}$-function, and replacing $\gamma$ with a $\Kcal_{[0,d)}$-function.

\begin{definition}(scNSS-Lyapunov Function) \label{def: scnss-lf}
    A size function $\mathcal{V}: \mathcal{S} \to \R_+$ is a scNSS-Lyapunov function if there exist $\alpha \in \mathcal{K}$ and $\gamma \in \mathcal{K}_{[0,d)}$ such that \eqref{eq: lyapnovFun} holds for all $\xi \in \mathcal{S}$ and all $\Theta \in \R^{m \times m}$ satisfying $\norm{\Theta \Theta^\top}<d$.
\end{definition}

By comparing the NSS-Lyapunov function in Definition \ref{def: nss_lf} with the scNSS-Lyapunov function in Definition \ref{def: scnss-lf}, we observe that the generator term $\mathcal{L}[\mathcal{V}](\xi,\Theta)$ may become unbounded as $\norm{\Theta \Theta^\top} \to d$ in the scNSS case, whereas it remains bounded for the NSS-Lyapunov function. The dissipative-type Lyapunov function for iNSS is introduced below. 
\begin{definition}(iNSS-Lyapunov Function)\label{def:integralnss_lf} 
    A size function $\mathcal{V}: \mathcal{S} \to \R_+$ is an iNSS-Lyapunov function if there exist $\alpha \in \mathcal{PD}$ and $\gamma \in \mathcal{K}$ such that \eqref{eq: lyapnovFun} holds for all $\xi \in \mathcal{S}$ and all $\Theta \in \R^{m \times m}$.
\end{definition}

The existence of a scNSS-Lyapunov function is sufficient to establish the scNSS property of a stochastic system. To facilitate the proof of the main result, we first introduce the following lemma.

\begin{lemma}\label{lm: probbeforeq1}
    Suppose there exists a scNSS-Lyapunov function $\Vcal$ for system \eqref{eq:nonLinearSys}. Let $c>1$, and assume $\normSup{\Sigma \Sigma^\top}  < \gamma^{-1} [\frac{1}{c}\sup_{r \in \R_+} \alpha(r)] =: d_1$. Define the set 
    \begin{align}
    \Dcal = \Big\{\xi \in \Scal|\, \Vcal(\xi) \le \alpha^{-1}\circ c\gamma(\normSup{\Sigma \Sigma^\top}) \Big\}
    \end{align}
    and the stopping time $q_1 = \inf\{t \ge 0|\,  \chi(t) \in \Dcal\}$. Then, for any $\epsilon \in (0,1)$, there exist $\beta_1 \in \Kcal \Lcal$ and $\gamma_1(r) = \alpha^{-1} \circ c\gamma(r) \in \Kcal_{[0,d_1)}$ such that 
    \begin{align}
        \Prob \Big\{ \Vcal(\chi(t \wedge q_1)) \le \beta_1(\Vcal(\chi_0), t) +  \gamma_1(\normSup{\Sigma \Sigma^\top}) \Big\}\ge 1-\epsilon/2, \, \forall t \ge 0.
    \end{align}
\end{lemma}
\begin{proof}
     Step 1: We establish that $\Vcal(\chi(t)) < \infty$ almost surely for all $t \in \R_+$. Let $\Scal_k = \{\xi \in \mathcal{S}|\, \Vcal(\xi) \le k \}$ for each integer $k \ge \Vcal(\chi_0)$. Fix the initial state $\chi_0 \in \Scal$, define the stopping time $\tau_k$ as the first time the sample path of the process $\{\chi(t)\}_{t \ge 0}$ exits $\Scal_k$, i.e.
    \begin{align}
        \tau_k = \inf\{ t \in \R_+| \, \Vcal(\chi(t)) > k \}.
    \end{align} 
    According to Dynkin's formula \cite[Lemma 3.2]{Book_khasminskii2012}, it gives
    \begin{align}\label{eq: barVcal_ItoInt}
        \E[\Vcal(\chi(t \wedge \tau_k))] &= \Vcal(\chi_0) + \E\bigg[\int_{0}^{t \wedge \tau_k} \Lcal[\Vcal](\chi(s), \Sigma(s)) \de s \bigg] , \, \forall t \in \R_+.
    \end{align}
    Due to the Lyapunov function in \eqref{eq: lyapnovFun}, it holds 
    \begin{align}\label{eq: barVcal_Diss}
    \begin{split}
         \E[\Vcal(\chi(t \wedge \tau_k))] &\le \Vcal(\chi_0) - \E\bigg[\int_{0}^{t \wedge \tau_k} \alpha(\Vcal(\chi(s))) \de s - \int_{0}^{t \wedge \tau_k} \gamma(\norm{\Sigma(s)\Sigma(s)^\top}) \de s  \bigg] \\
        &\le \Vcal(\chi_0) + t\gamma(\norm{\Sigma \Sigma^\top}_\infty), \, \forall t \in \R_+ .
    \end{split}
    \end{align}
    Using Chebyshev’s inequality \cite[p. 5]{Book_mao}, the probability that the process $\{ \chi(t \wedge \tau_k) \}_{t \ge 0 }$ stays within $\Scal_k$ is estimated as
    \begin{align}
        \Prob\big\{ \Vcal(\chi(t \wedge \tau_k)) < k \big\} \ge 1 - \frac{1}{k} \big(\Vcal(\chi_0) + t\gamma(\norm{\Sigma \Sigma^\top}_\infty) \big) , \, \forall t \in \R_+.
    \end{align}
    Since the condition $\Vcal(\chi(t \wedge \tau_k)) < k$ implies that $t < \tau_k$ (as otherwise, if $t \ge \tau_k$, then $\Vcal(\chi(t \wedge \tau_k)) = \Vcal(\chi(\tau_k)) = k$), it is concluded that
    \begin{align}
        \Prob\big\{ t < \tau_k \big\} \ge 1 - \frac{1}{k} \big(\Vcal(\chi_0) + t\gamma(\norm{\Sigma}_\infty) \big) , \, \forall t \in \R_+.
    \end{align}    
    Taking the limit $k \to \infty$, it follows that $\Prob\big\{  t < \lim_{k \to\infty}\tau_k \big\} = 1, \forall t \in \R_+$, which implies that $\Prob\big\{ \lim_{k \to \infty}\tau_k = \infty \big\} = 1$ and $\Prob\big\{ \Vcal(\chi(t)) < \infty \big\} = 1, \, \forall t \in \R_+$. 

    Step 2: we show that 
    \begin{align}\label{eq: barVcal}
        \E[\Vcal(\chi(t \wedge q_1))] \le \Vcal(\chi_0) - \E \bigg[\int_{0}^{t \wedge q_1 }  \alpha(\Vcal(\chi(s\wedge q_1))) - \gamma(\normSup{\Sigma \Sigma^\top}) \de s \bigg]. 
    \end{align}
    Since $\Prob\big\{ \lim_{k \to \infty}\tau_k = \infty \big\} = 1$, it follows that
    \begin{align}\label{eq: barVcal1}
    \begin{split}
        \E[\Vcal(\chi(t \wedge q_1))] &= \E \big\{ \Vcal[\chi( \liminf_{k\to\infty}(t \wedge q_1 \wedge \tau_k) )]\big\} = \E \big\{ \liminf_{k\to\infty} \Vcal[\chi( t \wedge q_1 \wedge \tau_k )]\big\} \\
        &\le \liminf_{k\to\infty}  \E \big\{ \Vcal[\chi( t \wedge q_1 \wedge \tau_k )]\big\},
    \end{split}
    \end{align}
    where the last inequality follows from Fatou's lemma \cite[p. 123]{Book_MaDonaldWeiss}. Plugging \eqref{eq: barVcal1} into \eqref{eq: barVcal_Diss} and applying the monotone convergence theorem \cite[p. 176]{Book_MaDonaldWeiss} result in \eqref{eq: barVcal}.

    Step 3: we demonstrate that for any $\epsilon \in (0,1)$, 
    \begin{align}\label{eq: Pstable}
        \Prob \bigg\{ \Vcal^+(\chi(t \wedge q_1)) \le \frac{2\Vcal^+(\chi_0)}{\epsilon} \bigg\} \ge 1 - \epsilon/2,
    \end{align}
    where $\Vcal^+: \Scal \to \R_+$ is defined as
    \begin{align}
        \Vcal^+(\xi) = \begin{cases} \Vcal(\xi)- \alpha^{-1} \circ c\gamma(\normSup{\Sigma \Sigma^\top}) ,\, &\text{if } \Vcal(\xi) \ge \alpha^{-1} \circ c\gamma(\normSup{\Sigma \Sigma^\top}) \\
        0,\, &\text{otherwise}.            
        \end{cases}
    \end{align}
    If $\chi_0 \in \Dcal$, then $q_1 = 0$ and $\chi(t \wedge q_1) = \chi_0$ for all $t\in \R_+$, so \eqref{eq: Pstable} holds trivially. If instead $\chi \in \Dcal^c$, note that since $t \wedge q_1 \le q_1$, the trajectory $\chi(t \wedge q_1)$ remains outside $\Dcal$ for all $t \in \R_+$. Consequently, by \eqref{eq: barVcal} and the definition of a scNSS-Lyapunov function in Definition \ref{def: scnss-lf},
    \begin{align}\label{eq:VchiExp}
        \E[\Vcal^+(\chi(t \wedge q_1))] \le \Vcal^+(\chi_0).
    \end{align}
    Combining \eqref{eq:VchiExp} with \eqref{eq: barVcal} and applying Chebyshev’s inequality \cite[p. 5]{Book_mao} establishes \eqref{eq: Pstable}.

    Step 4: we show that 
    \begin{align}\label{eq:asymptoticConv}
        \Prob \big\{ \lim_{t\to \infty}  \Vcal^+(\chi(t \wedge q_1)) = 0\big\} = 1.
    \end{align}
    It is clear that \eqref{eq:asymptoticConv} holds for $\chi_0 \in \Dcal$. Under the case of $\chi_0 \in \Dcal^c$, let $p_1 = \Prob\{\omega \in \Omega|\, q_1(\omega) < \infty \}$ and $p_2 = 1-p_1 = \Prob\{\omega \in \Omega|\, q_1(\omega) = \infty \}$. If $p_1=1$, then \eqref{eq:asymptoticConv} follows directly, since 
    $$
    \lim_{t \to \infty}\Vcal(\chi(t \wedge q_1)) = \Vcal(\chi(q_1)) = \alpha^{-1} \circ c\gamma(\normSup{\Sigma \Sigma^\top}).
    $$ 
    Now suppose that $p_1 < 1$. Then, it follows from \eqref{eq: barVcal} that
    \begin{align}
    \begin{split}
         \Vcal(\chi_0) &\ge \lim_{t\to \infty}\E \bigg[\int_{0}^{t \wedge q_1}  \alpha(\Vcal(\chi(s))) - \gamma(\normSup{\Sigma \Sigma^\top}) \de s \bigg] \\
         &\ge \lim_{t\to \infty}\E \bigg[\Idc_{\{\omega \in \Omega|\, q_1(\omega) = \infty \}}\int_{0}^{t}  \alpha(\Vcal(\chi(s))) - \gamma(\normSup{\Sigma \Sigma^\top}) \de s \bigg] \\
         &\ge \lim_{t \to \infty}\E \bigg[\Idc_{\{\omega \in \Omega|\, q_1(\omega) = \infty \}}\int_{0}^{t}   (c-1)\gamma(\normSup{\Sigma \Sigma^\top}) \de s \bigg]\\
         &= \lim_{t \to \infty}p_2(c-1)\gamma(\normSup{\Sigma \Sigma^\top}) t = \infty,
    \end{split}
    \end{align}
    which leads to a contradiction since $\Vcal(\chi_0) < \infty$. Hence, we must have $p_1=1$ and \eqref{eq:asymptoticConv} holds. 
    
    Step 5: we complete the proof. The relations in \eqref{eq: Pstable} and \eqref{eq:asymptoticConv} can be combined into
    \begin{align}\label{eq:GASconditions}
    \begin{split}
        &\Prob \bigg\{ \Vcal^+(\chi(t \wedge q_1)) \le \frac{2\Vcal^+(\chi_0)}{\epsilon} \text{ and } \lim_{t \to \infty} \Vcal^+(\chi(t \wedge q_1)) = 0 \bigg\} \ge 1 - \epsilon/2.
    \end{split}
    \end{align}
    The first condition in \eqref{eq:GASconditions} establishes the stability of the process $\chi(t \wedge q_1)$ with respect to the set $\Dcal$ under the topology induced by the size function $\Vcal$. Specifically, there exists a $\Kcal_\infty$-function $\delta_\epsilon(\epsilon_1) = \epsilon \epsilon_1/2$ such that, for any $\epsilon_1 \ge 0$, 
    \begin{align}
        \Vcal^+(\chi(t \wedge q_1)) \le \epsilon_1
    \end{align}
    whenever $\Vcal^+(\chi_0) \le \delta_\epsilon(\epsilon_1)$. The second condition in \eqref{eq:GASconditions} guarantees the attractivity of the process $\chi(t \wedge q_1)$ to the set $\Dcal$. That is, for any $r>0$ and $\epsilon_1>0$, there is a $T>0$, such that $\Vcal^+(\chi(t \wedge q_1)) \le \epsilon_1$ whenever $\Vcal^+(\chi_0) < r$ and $t \ge T$. Therefore, by \cite[Proposition 2.5]{Sontag_SIAM_1996}, there exists a function $\beta_1 \in \Kcal\Lcal$ such that
    \begin{align}
        \Prob \bigg\{ \Vcal^+(\chi(t\wedge q_1)) \le  \beta_1(\Vcal(\chi_0),t) \bigg\} \ge 1 - \epsilon/2. 
    \end{align}
    which concludes the proof with $\gamma_1(r) = \alpha^{-1} \circ c\gamma(r)$. 
\end{proof}

\begin{figure}
    \centering
    \includegraphics[width=0.4\linewidth]{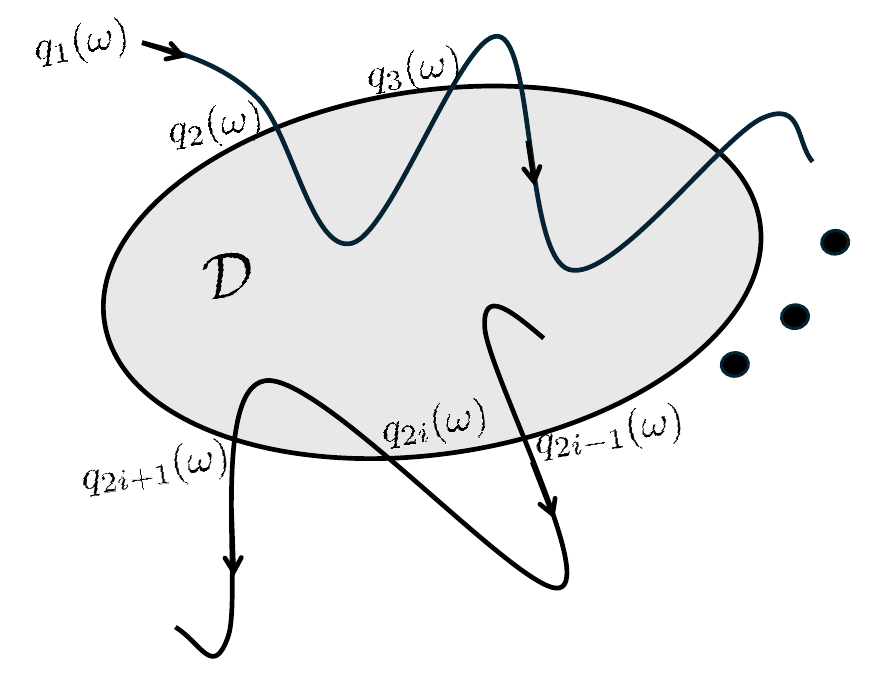}
    \caption{A sampled trajectory of the stochastic dynamics.}
    \label{fig:Stochastic Trajectory}
\end{figure}

We are now ready to state the main result of this section.
\begin{theorem}\label{thm: scNSS}
    System \eqref{eq:nonLinearSys} is scNSS if there exists a scNSS-Lyapunov function for it.
\end{theorem}

\begin{proof}
Assume the system \eqref{eq:nonLinearSys} admits a scNSS–Lyapunov function. Fix any constant $c > 1$, and let $d_1$ and $\Dcal$ be as defined in Lemma \ref{lm: probbeforeq1}. We divide the proof into two cases according to the initial condition $\chi_0$: $\chi_0 \in \Dcal^c$ and $\chi_0 \in \Dcal$. 

\textbf{Case 1 ($\chi_0 \in \Dcal^c$):} As shown in Fig. \ref{fig:Stochastic Trajectory}, let $q_1=0$, and for all $i \in \mathbb{Z}_+$, define $q_{2i}$ and $q_{2i+1}$ as the sequences of stopping times corresponding to the $i$th entry into and exit from the set $\Dcal$, respectively. That is 
\begin{align}\label{eq: entryExitTime}
\begin{split}
    q_{2i} &= \begin{cases}
        \inf \{t > q_{2i-1}|\, \chi(t) \in  \Dcal\} &\text{ if } \{t > q_{2i-1}|\, \chi(t) \in  \Dcal\} \neq \emptyset \\
        \infty &\text{ otherwise }
    \end{cases} \\
    q_{2i+1} &= \begin{cases}
        \inf \{t > q_{2i}|\, \chi(t) \in  \Dcal^c\backslash \partial \Dcal\} &\text{ if } \{t > q_{2i}|\, \chi(t) \in  \Dcal^c\backslash \partial \Dcal \} \neq \emptyset \\
        \infty &\text{ otherwise }
    \end{cases}   
\end{split}
\end{align}
By the definition of scNSS-Lyapunov function, the following properties hold. For any fixed $t \ge 0$, under the event $\mathcal{A}_{2i-1}(t) = \{\omega \in \Omega|\, t \in [q_{2i-1}(\omega), q_{2i}(\omega))\}$ ($i \in \mathbb{Z}_+$), the states $\chi(\omega, t)$ are outside of $\Dcal$ and it follows that
\begin{align}\label{eq:LVoutsideS}
\begin{split}
    \Lcal[\Vcal](\chi(t), \Sigma(t)) &\le - \alpha(\Vcal(\chi(t))) + \gamma(\Sigma(t)\Sigma(t)^\top) \\
    &\le -\big( 1 - {1}/{c}\big) c \gamma(\normSup{\Sigma \Sigma^\top}).
\end{split}
\end{align}
For any fixed $t \ge 0$, under the event $\mathcal{A}_{2i}(t) = \{\omega \in \Omega|\, t \in [q_{2i}(\omega), q_{2i+1}(\omega))\}$ ($i \in \mathbb{Z}_+$), the states $\chi(\omega, t)$ are inside of $\Dcal$ and 
\begin{align}\label{eq: insideIneq}
    \Vcal(\chi(t)) \le \alpha^{-1} \circ c\gamma(\normSup{\Sigma \Sigma^\top}).
\end{align}

We will first analyze the stochastic process $\chi((t \vee q_{2i-1}) \wedge q_{2i})$, which lies outside the set $\Dcal$. By \eqref{eq:ItoIntegral}, it holds: 
\begin{align}\label{eq: twoItos}
\begin{split}
    \Vcal(\chi(t \wedge q_{2i-1})) &= \Vcal(\chi_0) + \int_{0}^{t \wedge q_{2i-1}} \Lcal[\Vcal](\chi(s), \Sigma(s)) \de s + \int_{0}^{t \wedge q_{2i-1}}  \innprod{\nabla \mathcal{V}(\chi(s))}{g(\chi(s)) \Sigma(s) \de B(s)}\\
    \Vcal(\chi(t \wedge q_{2i})) &= \Vcal(\chi_0) + \int_{0}^{t \wedge q_{2i}} \Lcal[\Vcal](\chi(s), \Sigma(s)) \de s + \int_{0}^{t \wedge q_{2i}}  \innprod{\nabla \mathcal{V}(\chi(s))}{g(\chi(s)) \Sigma(s) \de B(s)}
\end{split}
\end{align}
Taking the difference between the two equations in \eqref{eq: twoItos} results in 
\begin{align}
\begin{split}
     &\Vcal(\chi(t \wedge q_{2i})) - \Vcal(\chi(t \wedge q_{2i-1})) = \int_{t \wedge q_{2i-1}}^{t \wedge q_{2i}} \Lcal[\Vcal](\chi(s), \Sigma(s)) \de s \\
     &\quad + \int_{t \wedge q_{2i-1}}^{t \wedge q_{2i}}  \innprod{\nabla \mathcal{V}(\chi(s))}{g(\chi(s)) \Sigma(s) \de B(s)},
\end{split}
\end{align}
which can be rewritten as
\begin{align}
\begin{split}
    &\Vcal(\chi((t \vee q_{2i-1}) \wedge q_{2i})) - \Vcal(\chi(q_{2i-1})) =  \int_{q_{2i-1}}^{(t \vee q_{2i-1}) \wedge q_{2i}} \Lcal[\Vcal](\chi(s), \Sigma(s)) \de s \\
    &\quad + \int_{q_{2i-1}}^{(t \vee q_{2i-1}) \wedge q_{2i}}  \innprod{\nabla \mathcal{V}(\chi(s))}{g(\chi(s)) \Sigma(s) \de B(s)}.
\end{split}
\end{align}
Since $\Lcal[\Vcal](\chi(s), \Sigma(s)) \le 0$ for all $s \in [q_{2i-1}, (t \vee q_{2i-1}) \wedge q_{2i}]$, $\Vcal(\chi((t \vee q_{2i-1}) \wedge q_{2i}))$ is a supermartingale, which yields
\begin{align}\label{eq:boundoutside}
    \E[\Vcal(\chi((t \vee q_{2i-1}) \wedge q_{2i}))] \le \E[\Vcal(\chi(q_{2i-1}))] = \alpha^{-1} \circ c\gamma(\normSup{\Sigma \Sigma^\top}).
\end{align}

We now present the main result. For any $\omega \in \Omega$, the time axis admits the disjoint decomposition
$$
[0,\infty) = \cup_{i=1}^\infty [q_i(\omega), q_{i+1}(\omega)).
$$ 
Fix $t \ge 0$. Then there exists $\bar{i}(t,\omega) \in \mathbb{Z}_+$, such that $t \in [q_{\bar{i}}(\omega), q_{\bar{i}+1}(\omega))$.  By construction, this implies that for any $\omega \in \Omega$
$$
\omega \in \Acal_{\bar{i}}(t) = \{\omega \in \Omega|\, t \in [q_{\bar{i}}(\omega), q_{\bar{i}+1}(\omega))\} \subset \cup_{i=1}^\infty \Acal_i(t).
$$
Moreover, the events $\Acal_{i}(t)$ and $\Acal_{j}(t)$ are mutually disjoint for $i \neq j$. Consequently,
$$
\Omega = \sum_{i=1}^\infty \Acal_i(t).
$$
It follows that
\begin{align}\label{eq: Vafterq1}
\begin{split}
    &\E[\Vcal(\chi(t\vee q_2))] =  \E[\Vcal(\chi(t\vee q_2))\Idc_{\Acal_{1}(t)}(\omega)]  + \sum_{i=1}^\infty \E[\Vcal(\chi(t\vee q_2))\Idc_{\Acal_{2i}(t)}(\omega)] \\
    &\quad + \sum_{i=2}^\infty \E[\Vcal(\chi(t\vee q_2))\Idc_{\Acal_{2i-1}(t)}(\omega)]\\
    &\quad \le \alpha^{-1}\circ c\gamma(\normSup{\Sigma \Sigma^\top}) \bigg[ \Prob\{\Acal_{1}(t)\} +  \sum_{i=1}^\infty\Prob\{\Acal_{2i}(t)\} \bigg] +  \sum_{i=2}^\infty \E[\Vcal(\chi(t))\Idc_{\Acal_{2i-1}(t)}(\omega)]
\end{split}
\end{align}
where the inequality follows from \eqref{eq: insideIneq}. The last term in \eqref{eq: Vafterq1} can be bounded as follows:
\begin{align}\label{eq: Vexpoutside}
\begin{split}
    \E[\Vcal(\chi(t))\Idc_{\Acal_{2i-1}(t)}(\omega)] &= \E[\Vcal(\chi((t \vee q_{2i-1})\wedge q_{2i}))\Idc_{\Acal_{2i-1}(t)}(\omega)] \\
    &=\E[\Vcal(\chi((t \vee q_{2i-1})\wedge q_{2i}))] - \E[\Vcal(\chi((t \vee q_{2i-1})\wedge q_{2i}))\Idc_{\Acal_{2i-1}(t)^c}(\omega)] \\
    &=\E[\Vcal(\chi((t \vee q_{2i-1})\wedge q_{2i}))] - \alpha^{-1}\circ c\gamma(\normSup{\Sigma \Sigma^\top}) \Prob\{\Acal_{2i-1}(t)^c\} \\
    &\le  \alpha^{-1}\circ c\gamma(\normSup{\Sigma \Sigma^\top})\Prob\{\Acal_{2i-1}(t)\}
\end{split}
\end{align}
where the third line uses the fact that $\Vcal(\chi(q_{2i-1})) = \Vcal(\chi(q_{2i})) = \alpha^{-1}\circ c\gamma(\normSup{\Sigma})$, and the last inequality follows from \eqref{eq:boundoutside}. Substituting \eqref{eq: Vexpoutside} into \eqref{eq: Vafterq1}, we obtain
\begin{align}
    \E[\Vcal(\chi(t\vee q_2))] \le \alpha^{-1}\circ c\gamma(\normSup{\Sigma \Sigma^\top})
\end{align}
which, by Chebyshev’s inequality \cite[p. 5]{Book_mao}, implies
\begin{align}\label{eq: probafterq1}
    \Prob\bigg\{ \Vcal(\chi(t\vee q_2)) \le \frac{2}{\epsilon} \alpha^{-1}\circ c\gamma(\normSup{\Sigma \Sigma^\top})\bigg\} \ge 1 - \epsilon/2.
\end{align}
Combining \eqref{eq: probafterq1} with Lemma~\ref{lm: probbeforeq1}, we obtain
\begin{align}\label{eq:NSSxoutD}
\begin{split}
    &\Prob\big\{ \Vcal(\chi(t)) \le \beta_1(\Vcal(\chi_0),t) + \frac{2}{\epsilon} \alpha^{-1}\circ c\gamma(\normSup{\Sigma \Sigma^\top})\big\} \\
    &\ge \Prob\bigg\{ [t\le q_2] \cap \big[ \Vcal(\chi(t \wedge q_2)) \le \beta_1(\Vcal(\chi_0),t) + \frac{2}{\epsilon} \alpha^{-1}\circ c\gamma(\normSup{\Sigma \Sigma^\top})\big] \\
    &\quad \cup [t > q_2] \cap \big[\Vcal(\chi(t \vee q_2)) \le \beta_1(\Vcal(\chi_0),t) + \frac{2}{\epsilon} \alpha^{-1}\circ c\gamma(\normSup{\Sigma \Sigma^\top})\big] \bigg\} \\
    &=1 - \Prob\bigg\{ [t> q_2] \cap \big[ \Vcal(\chi(t \vee q_2)) > \beta_1(\Vcal(\chi_0),t) + \frac{2}{\epsilon} \alpha^{-1}\circ c\gamma(\normSup{\Sigma \Sigma^\top})\big] \\
    &\quad \cup [t \le q_2] \cap \big[\Vcal(\chi(t \wedge q_2)) > \beta_1(\Vcal(\chi_0),t)+\frac{2}{\epsilon} \alpha^{-1}\circ c\gamma(\normSup{\Sigma \Sigma^\top})\big] \bigg\} \\
    &\ge 1 - \epsilon 
\end{split}
\end{align}
which completes the proof for $\chi_0 \in \Dcal^c$.

\textbf{Case 2 ($\chi_0 \in \Dcal$):} In this case, $q_2 = 0$ almost surely. Therefore, by \eqref{eq: probafterq1}, it follows that
\begin{align}\label{eq:NSSxinD}
    \Prob\bigg\{ \Vcal(\chi(t)) \le \frac{2}{\epsilon} \alpha^{-1}\circ c\gamma(\normSup{\Sigma \Sigma^\top})\bigg\} \ge 1 - \epsilon/2.
\end{align}
In conclusion, when $\normSup{\Sigma \Sigma^\top} < d_1$, 
$$
\Prob\big\{ \Vcal(\chi(t)) \le \beta_1(\Vcal(\chi_0),t) + \frac{2}{\epsilon} \alpha^{-1}\circ c\gamma(\normSup{\Sigma \Sigma^\top})\big\} \ge 1-\epsilon
$$
holds in both cases. Therefore, the stochastic system is scNSS.
\end{proof}

The existence of an NSS-Lyapunov function implies that system \eqref{eq:nonLinearSys} is NSS. This result can be viewed as a special case of the scNSS property.
\begin{theorem}\label{thm: NSS}
    System \eqref{eq:nonLinearSys} is NSS if there exists a NSS-Lyapunov function for it.
\end{theorem}
\begin{proof}
    The existence of an NSS-Lyapunov function directly implies that the constant $d_1$ defined in Lemma \ref{lm: probbeforeq1} satisfies 
    $$
    d_1 =\gamma^{-1} [\frac{1}{c}\sup_{r \in \R_+} \alpha(r)] = \infty.
    $$
   Consequently, by the proof of Theorem~\ref{thm: scNSS}, both \eqref{eq:NSSxoutD} and \eqref{eq:NSSxinD} hold for any $\normSup{\Sigma \Sigma^\top} < d_1 = \infty$. This establishes the NSS property of the dynamical system.
\end{proof}

As shown in \cite[Theorem 2]{Ito2020}, the existence of an iNSS-Lyapunov function guarantees that system \eqref{eq:nonLinearSys} is iNSS.
\begin{theorem}[\cite{Ito2020}]\label{thm: iNSS}
    System \eqref{eq:nonLinearSys} is iNSS if there exists an iNSS-Lyapunov function for it.
\end{theorem}

In the next section, we apply the developed notions of NSS to analyze the robustness of a specific class of stochastic dynamics: stochastic gradient dynamics.

\section{NSS for Stochastic Gradient Dynamics}\label{sec:NSSStochasticGradient}
In this section, we apply different notions of NSS to analyze gradient flows under stochastic noise with time-varying variance.

\subsection{Robustness of Overdamped Langevin Diffusion}
Gradient flows are efficient for solving the optimization problem
\begin{align}\label{eq: optimization}
    \min_{z \in \Zcal}\Jcal(z) 
\end{align}
where $z \in \R^n$ is the decision variable and $\Zcal \subset \R^n$ is a feasible set that is diffeomorphic to $\R^n$. By continuously updating the decision variable in the direction of gradient descent
\begin{align}\label{eq:gradientFLow}
    \de z(t) = - \nabla \Jcal(z(t)) \de t,
\end{align}
the algorithm drives the decision variables toward the set of critical points $\Zcal_c = \{z \in \Zcal| \nabla \Jcal(z) = 0 \} $ under mild regularity conditions:
\begin{assumption}\label{ass: regularity}
    The objective function $\Jcal$ is twice continuously differentiable, bounded below, and coercive on the set $\Zcal$. 
\end{assumption}
Moreover, in the deterministic case, Assumption \ref{ass: regularity} helps rule out convergence to strict saddle points or local maxima \cite{Arthur2024convergence,Arthur2024CDC}, thereby ensuring convergence to a local minimum. To guarantee convergence to the global minimum $\Jcal^*$ of problem \eqref{eq: optimization}, additional assumptions are required--most notably, convexity of $\Jcal$. Another sufficient condition, which ensures exponential convergence to the global minimum, was proposed in \cite{Lojasiewicz1963,Polyak1963}:
\begin{definition}(Polyak-\L ojasiewicz (PL) Condition)
    The function $\Jcal$ satisfies the Polyak-\L ojasiewicz (PL) condition if there exists $c > 0$ such that
    \begin{align} \label{eq: PLInequality}
        \norm{\nabla \Jcal(z)} \ge \mu(\Jcal(z) - \Jcal^*),
    \end{align}
    where $\mu(r) = \sqrt{cr}$ for all $r \in \R_+$.
\end{definition}

While convergence is a key consideration, robustness is another critical factor that determines the efficacy of a gradient algorithm. In the presence of noise, the algorithm should still converge to a solution that is close to the optimum. As noted by \cite{book_Polyak}, many disturbances affecting gradient descent methods can be modeled as random noise. In such cases, the gradient dynamics can be represented by a stochastic differential equation, often referred to as the \emph{overdamped Langevin diffusion}:
\begin{align}\label{eq: gradientDynamics}
    \de z(t) = -\nabla \Jcal(z(t)) \de t + G(z(t)) \Sigma(t) \de B(t)
\end{align}
where $B$ is a standard, independent $n$-dimensional Brownian motion, $G: \Zcal \to \R^{n \times n}$ is locally Lipschitz continuous and locally bounded, and $\Sigma: \R_+ \to \R^{n \times n}$ is Borel measurable and locally essentially bounded, which characterizes the intensity of the stochastic perturbations. Under persistent stochastic noise, the gradient dynamics can never remain exactly at a critical point where $\nabla \Jcal(z)=0$. Instead, we aim for the algorithm to remain within a neighborhood of the optimum with high probability. The size of this neighborhood is governed by the intensity of the noise, characterized by $\normSup{\Sigma \Sigma^\top}$. Formally, this requirement can be expressed by demanding that the stochastic gradient dynamics exhibit NSS.

We aim to establish a connection between the PL condition and the NSS of stochastic gradient dynamics. The classical PL condition requires $\mu$ to be the square-root function, which can be restrictive and limit its applicability in broader contexts where such a strict condition may not hold. To address this limitation, we introduce relaxed versions of the PL condition by employing comparison functions.

\begin{definition}($\Kcal_\infty$-PL Condition)
    The function $\Jcal$ satisfies the $\Kcal_\infty$-PL condition if there exists a $\mu \in \Kcal_\infty$ such that \eqref{eq: PLInequality} holds.
\end{definition}
The $\Kcal_\infty$-PL condition can be further relaxed by removing the unboundedness requirement of the $\Kcal_\infty$-function.
\begin{definition}($\Kcal$-PL Condition)
    The function $\Jcal$ satisfies the $\Kcal$-PL condition if there exists a $\mu \in \Kcal$ such that
    such that \eqref{eq: PLInequality} holds.   
\end{definition}
The $\Kcal$-PL condition can be further relaxed by removing the monotonicity requirement of the $\Kcal$-function.
\begin{definition}($\mathcal{PD}$-PL Condition)
    The function $\Jcal$ satisfies the $\mathcal{PD}$-PL condition if there exists a $\mu \in \mathcal{PD}$ such that \eqref{eq: PLInequality} holds.   
\end{definition}

\subsubsection{Robustness Analysis with Global Lipchitz Condition}
The following main theorem establishes the connection between NSS and the different variants of the PL conditions by assuming that $\nabla \Jcal$ is globally Lipchitz continuous.
\begin{theorem}\label{thm:overLangevinLipchitz}
    Let the objective function $\Jcal$ have a global $L$-Lipschitz continuous gradient, and let $G(z)$ be globally bounded, i.e., $\normF{G} \le K_G$. Then:
    \begin{enumerate}
        \item If the function $\Jcal$ satisfies the $\Kcal_\infty$-PL condition, the overdamped Langevin diffusion \eqref{eq: gradientDynamics} is NSS;
        \item If the function $\Jcal$ satisfies the $\Kcal$-PL condition, the overdamped Langevin diffusion \eqref{eq: gradientDynamics} is scNSS and iNSS.
        \item If the function $\Jcal$ satisfies the $\mathcal{PD}$-PL condition, the overdamped Langevin diffusion \eqref{eq: gradientDynamics} is iNSS.
    \end{enumerate}
\end{theorem}
\begin{proof}
    Assumption \ref{ass: regularity} ensures that $\Jcal - \Jcal^*$ is a size function. The $L$-Lipschitz continuity of $\nabla \Jcal(z)$ implies that $\norm{\nabla^2\Jcal(z)} \le L$ \cite[Lemma 1.2.2]{nesterov2013introductory}. We choose $\Jcal - \Jcal^*$ as a candidate Lyapunov function, and applying the infinitesimal generator $\Lcal$ to it yields:
    \begin{align}\label{eq: LJoperator}
    \begin{split}
        \Lcal[\Jcal](z(t), \Sigma(t)) &= -\norm{\nabla \Jcal(z(t))}^2 + \frac{1}{2}\innprod{G(z(t))\Sigma(t)}{\nabla^2 \Jcal(z(t)) G(z(t))\Sigma(t)} \\
        &\le -\mu(\Jcal(z(t)) - \Jcal^*)^2 + \frac{1}{2}LK_G^2\norm{\Sigma(t)\Sigma(t)^\top}
    \end{split}
    \end{align}
    where the inequality follows from a suitable variant of the PL conditions and Lemma \ref{lm:traceIneq}. When $\mu$ is a $\Kcal_\infty$-function, $\Jcal - \Jcal^*$ qualifies as an NSS-Lyapunov function as defined in Definition \ref{def: nss_lf},  which, by Theorem \ref{thm: NSS}, implies that the gradient dynamics \eqref{eq: gradientDynamics} is NSS. 

    When $\mu$ is a $\Kcal$-function, $\Jcal - \Jcal^*$ is a scISS-Lyapunov function as defined in Definition \ref{def: scnss-lf}, and gradient dynamics \eqref{eq: gradientDynamics} is scNSS by Theorem \ref{thm: scNSS}. 

     When $\mu$ is a $\mathcal{PD}$-function, $\Jcal - \Jcal^*$ is an iISS-Lyapunov function as defined in Definition \ref{def:integralnss_lf}, and gradient dynamics \eqref{eq: gradientDynamics} is iNSS by Theorem \ref{thm: iNSS}.
\end{proof}

\subsubsection{Robustness Analysis without Global Lipchitz Condition}
In the preceding discussion, the global Lipschitz continuity of $\nabla \Jcal(z)$ is required to establish the connection between various versions of NSS and the generalized PL conditions. However, this global assumption may limit practical applicability. To address this, we relax the global Lipschitz requirement by appropriately tuning the learning rate in the gradient dynamics. Specifically, we allow for a state-dependent learning rate $\eta(\Jcal(z))>0$ in the dynamics:
\begin{align}\label{eq: gradientwithTuning}
    \de z(t) = -\eta(\Jcal(z(t)))\nabla \Jcal(z(t)) \de t + G(z(t))\Sigma(t) \de B(t)
\end{align}

Since $\Jcal$ is twice continuously differentiable, its gradient is locally Lipschitz and its Hessian is locally bounded. To facilitate the robustness analysis, define the sublevel set
\begin{align}\label{eq:Zsublevelset}
\Zcal_h = \{z \in \Zcal|\, \Jcal(z) - \Jcal^* \le h \}.
\end{align}
Note that for any $z \in \Zcal$, we have $z \in \Zcal_{\Jcal(z)-\Jcal^*}$. Over the sublevel set, define
\begin{align}\label{eq:Hessensup}
\bar{L}(h) = \frac{1}{2}\max_{z \in \Zcal_h} \norm{\nabla^2 \Jcal(z)} \normF{G(z)}^2
\end{align}
which is a continuous and nondecreasing function of $h$. Then, the function 
\begin{align}
\tilde{L}(h) = \bar{L}(h) - \bar{L}(0)
\end{align} 
can be assumed to be of $\Kcal_\infty$, since otherwise it can be modify by adding an identity function, ensuring that $\tilde{L}(h)+h \in \Kcal_\infty$.

\begin{theorem}\label{thm:NSSwithoutLipchitz}
    Suppose $\Jcal(z)$ satisfies the $\Kcal$-PL condition. Then,  
    \begin{enumerate}
        \item if $\lim_{h \to \infty } \frac{\tilde{L}(h)}{\eta(h)\mu(h)^2} = d_2>0$, the overdamped Langevin diffusion in \eqref{eq: gradientwithTuning} is scNSS;
        \item if $\lim_{h \to \infty } \frac{\tilde{L}(h)}{\eta(h)\mu(h)^2} = 0$, the  overdamped Langevin diffusion in \eqref{eq: gradientwithTuning} is NSS.
    \end{enumerate}
\end{theorem}
\begin{proof}
In the proof, we assume without loss of generality that $\Jcal^*=0$ and omit the time index $t$ for simplicity of notation. From \eqref{eq: LJoperator} and \eqref{eq:Hessensup}, it follows that:
\begin{align}\label{eq:LJwithStep}
\begin{split}
    \Lcal[\Jcal](z, \Sigma) &\le - \eta(\Jcal(z))\mu(\Jcal(z))^2 + \frac{1}{2}\norm{\nabla \Jcal(z)} \normF{G(z)}^2 \norm{\Sigma\Sigma^\top} \\
    &\le - \eta(\Jcal(z))\mu(\Jcal(z))^2 +  (\bar{L}(0) + \tilde{L}(\Jcal(z))) \norm{\Sigma\Sigma^\top}.
\end{split}
\end{align}    

Define
\begin{align}\label{eq:m0Define}
    m_1(h) = \inf_{r \ge h} \frac{\eta(r) \mu(r)^2}{ \tilde{L}(r)} \le \frac{ \eta(h) \mu(h)^2}{\tilde{L}(h)}
\end{align}
which is continuous and nondecreasing. If $\lim_{h \to \infty } \frac{\tilde{L}(h)}{\eta(h)\mu(h)^2} = d_2$, then $\lim_{h \to \infty } m_1(h) = {1}/{d_2}$. Consequently, there exists a $\Kcal$-function $m_2$ with range $[0,1/d_2)$ such that $m_1(h) \ge m_2(h)$ for all $h \ge 0$. Define $m_3 = \frac{1}{2}m_2\circ \tilde{L}^{-1}$ which is of class $\Kcal$ and with range $[0, {1}/{(2d_2)})$. For any signal $\Sigma$ with $\normSup{\Sigma \Sigma^\top} < 1/(2d_2)$, it follows from \eqref{eq:LJwithStep} that
\begin{align}\label{eq:LJwithStep2}
\begin{split}
    &\Lcal[\Jcal](z(t), \Sigma(t)) \le  - \eta(\Jcal(z))\mu(\Jcal(z))^2 + (\bar{L}(0) + m_3^{-1}(\norm{\Sigma \Sigma^\top})) \norm{\Sigma \Sigma^\top} + \tilde{L}(\Jcal(z))m_3( \tilde{L}(\Jcal(z))).
\end{split}
\end{align}
Since $\tilde{L}(h) m_2(h) \le \eta(h) \mu(h)^2$ for any $h\ge0$, \eqref{eq:LJwithStep2} can be rewritten as
\begin{align}
     \Lcal[\Jcal](z(t), \Sigma(t)) &\le - \frac{1}{2} \eta(\Jcal(z))\mu(\Jcal(z))^2 + (\bar{L}(0) + m_3^{-1}(\norm{\Sigma \Sigma^\top})) \norm{\Sigma \Sigma^\top}.
\end{align}
By \eqref{eq:m0Define}, we have $\eta(h)\mu(h)^2 \ge m_2(h)\tilde{L}(h),\, \forall h \ge 0 $. Since $m_2$ is of class $\Kcal$ and $\tilde{L}$ is of class $\Kcal_\infty$, it follows that $\eta(h)\mu(h)^2$ is lower bounded by a $\Kcal$-function. In addition, since $m_3$ is a $\Kcal$-function with range $[0,1/(2d_2))$, the function $\rho(h) = (\bar{L}(0) + m_3^{-1}(h))h$ belongs to $\Kcal_{[0, 1/(2d_2))}$. Therefore, according to Definition~\ref{def: scnss-lf}, $\Jcal$ is a scNSS-Lyapunov function. Consequently, by Theorem~\ref{thm: scNSS}, the stochastic system \eqref{eq: gradientwithTuning} is scNSS.

If $\lim_{h \to \infty } \frac{\tilde{L}(h)}{\eta(h)\mu(h)^2} = 0$, then $\lim_{h \to \infty } m_1(h) = \infty$. Hence, $m_1$ can be lower bounded by a $\Kcal_\infty$-function $m_2$, and $m_3 = \frac{1}{2}m_2\circ \tilde{L}^{-1}$ is of class $\Kcal_\infty$. Since $\eta(h)\mu(h)^2 \ge m_2(h)\tilde{L}(h)$, where $m_2(\cdot)\tilde{L}(\cdot)$ is a $\Kcal_\infty$-function, and $\rho(h) = (\bar{L}(0) + m_3^{-1}(h))h$ is of class $\Kcal_\infty$, according to Definition~\ref{def: nss_lf}, $\Jcal$ is a NSS-Lyapunov function. Consequently, by Theorem~\ref{thm: NSS}, the stochastic system \eqref{eq: gradientwithTuning} is NSS.
\end{proof}

As a concrete example, one may choose $\eta(h) = \tilde{L}(h)$ such that $\lim_{h \to \infty } \frac{\tilde{L}(h)}{\eta(h)\mu(h)^2} = \lim_{h \to \infty } \frac{1}{\mu(h)^2} = d_2$. Then, the stochastic gradient dynamics is scNSS. If $\eta(h) = h\tilde{L}(h)$, $\lim_{h \to \infty } \frac{\tilde{L}(h)}{\eta(h)\mu(h)^2} = 0$ and the stochastic gradient dynamics is NSS.

\subsection{Robustness of Underdamped Langevin Diffusion }
The heavy ball optimization algorithm was derived from the physical analogy of a ball moving in the potential field $\Jcal(z)$ under the influence of friction~\cite{Polyak1964HeavyBall, polyak2017lyapunov}. It modifies standard gradient descent by incorporating a momentum term that accelerates convergence. The heavy ball method can be viewed as the Euler discretization of the following second-order differential equation:
\begin{align}\label{eq:heavyballFlow}
\begin{split}
    \de{z}(t) &= v(t) \de t, \\
    \de{v}(t) &= -\eta \nabla \Jcal(z(t))\de t - c v(t)\de t,
\end{split}
\end{align}
where $\eta>0$ is a learning-rate parameter and $c>0$ is the damping coefficient. It is noted that the gradient flow in~\eqref{eq:heavyballFlow} can be obtained as a singular perturbation of~\eqref{eq:gradientFLow} in the limit $c \to \infty$. To see this, let $\epsilon = \tfrac{1}{c}$ and rewrite~\eqref{eq:heavyballFlow} as
\begin{align}
\begin{split}
    \dot{z}(t) &= v(t), \\
    \epsilon \dot{v}(t) &= - \tfrac{\eta}{c}\,\nabla \Jcal(z(t)) - v(t).    
\end{split}
\end{align}
Here, $v(t)$ evolves on the fast time scale and rapidly converges to its quasi-steady state $\bar{v}(z) = -\frac{\eta}{c}\nabla \Jcal(z)$. By Tikhonov's theorem~\cite[Theorem~11.1]{book_Khalil}, we obtain 
\[
    z(t) - \bar{z}(t) = O(\epsilon),
\] 
where $\bar{z}(t)$ satisfies the reduced (slow) dynamics
\begin{align}
    \dot{\bar{z}}(t) = -\tfrac{\eta}{c} \nabla \Jcal(\bar{z}(t)).
\end{align}
This is precisely the standard first-order gradient flow in \eqref{eq:gradientFLow} with a learning rate $\eta_1 = \eta/c$.

As noted by~\cite{book_Polyak}, many disturbances affecting gradient descent methods can be modeled as random noise. In such cases, the gradient dynamics in \eqref{eq:heavyballFlow} can be represented by a stochastic differential equation, often referred to as the \emph{underdamped Langevin diffusion}:
\begin{align}\label{eq:HeavyballFlowStochastic}
\begin{split}
    \de z(t) &= v(t) \de t\\
    \de v(t) &= -\eta \nabla \Jcal(z(t)) \de t - cv(t) \de t  + G(z(t), v(t)) \Sigma(t)\de B(t).
\end{split}
\end{align}

\subsubsection{Robustness Analysis with Global Lipchitz Condition}
If the objective has a globally Lipschitz-continuous gradient, we can connect generalized PL conditions to NSS for the underdamped Langevin diffusion.

\begin{theorem}\label{thm:DampedLangevinLipchitz}
    Let the objective function $\Jcal$ have a global $L$-Lipschitz continuous gradient, and let $G(z,v)$ be globally bounded, i.e., $\normF{G} \le K_G$. Then:
    \begin{enumerate}
        \item If the function $\Jcal$ satisfies the $\Kcal_\infty$-PL condition, the underdamped Langevin diffusion \eqref{eq:HeavyballFlowStochastic} is NSS;
        \item If the function $\Jcal$ satisfies the $\Kcal$-PL condition, the underdamped Langevin diffusion \eqref{eq:HeavyballFlowStochastic} is scNSS and iNSS.
        \item If the function $\Jcal$ satisfies the $\mathcal{PD}$-PL condition, the underdamped Langevin diffusion \eqref{eq:HeavyballFlowStochastic} is iNSS.
    \end{enumerate}
\end{theorem}
\begin{proof}
We use the mixed Lyapunov candidate
\begin{align}\label{eq:LyapunovHeavyBall}
    \Vcal_2(z,v) = \Jcal(z) - \Jcal^* +  \lambda_1\innprod{v}{\nabla \Jcal(z)} + \frac{\lambda_2}{2}\innprod{v}{v}
\end{align}
where the weights $\lambda_1$ and $\lambda_2$ are chosen as
\begin{align}
\begin{split}
    &0 < \lambda_1 \le \min\bigg\{\frac{1}{2\eta + c}, \frac{1}{2L}, \frac{c}{2(\eta L + c^2)} \bigg\}, \\
    &\lambda_2 = \frac{1-\lambda_1 c}{\eta}.
\end{split}
\end{align}
The $L$-smoothness of $\nabla \Jcal$ gives 
\begin{align}\label{eq:Lsmoothness}
\Jcal(z) - \Jcal^* \ge \frac{1}{2L}\innprod{\nabla \Jcal(z)}{\nabla \Jcal(z)}.
\end{align}
Plugging \eqref{eq:Lsmoothness} into \eqref{eq:LyapunovHeavyBall} and using Young's inequality, we can ensure $\Vcal_2$ is a size function over $\Zcal \times \R^n$ by
\begin{align}
\begin{split}
    \Vcal_2(z,v) &\ge \frac{1}{2}(\Jcal(z) - \Jcal^*) + \frac{1}{4L} \innprod{\nabla \Jcal(z)}{\nabla \Jcal(z)} + \lambda_1\innprod{v}{\nabla \Jcal(z)}  + \frac{\lambda_2}{2}\innprod{v}{v} \\
    &\ge \frac{1}{2}(\Jcal(z) - \Jcal^*) + \frac{\lambda_2}{4}\innprod{v}{v}. 
\end{split}
\end{align}

Along the trajectories of \eqref{eq:HeavyballFlowStochastic}, and under the restrictions on the weights $\lambda_1$ and $\lambda_2$, the generator $\Lcal$ applied to $\Vcal_2$ satisfies
\begin{align}
\begin{split}
    \Lcal[\Vcal_2](z,v) &= \innprod{\nabla \Jcal (z)}{ v}  + \lambda_1\innprod{v}{\nabla^2 \Jcal(z) v} - \lambda_1 \eta \innprod{\nabla \Jcal(z)}{\nabla \Jcal(z)}   \\
    &\quad -\lambda_1 c \innprod{\nabla \Jcal(z)}{v}  - \lambda_2 \eta \innprod{v}{\nabla \Jcal(z)}  - \lambda_2 c\innprod{v}{v} + \frac{\lambda_2}{2}  \innprod{G\Sigma}{G\Sigma} \\
    &\le - \lambda_1 \eta \innprod{\nabla \Jcal(z)}{\nabla \Jcal(z)}  - \frac{c}{2\eta} \innprod{v}{v} + \frac{\lambda_2}{2}K_G^2 \norm{\Sigma \Sigma^\top} \\
    &\le - \lambda_1 \eta \mu (\Jcal(z) - \Jcal^*)^2 - \lambda_1 \eta \frac{c}{2\lambda_1\eta^2} \innprod{v}{v} + \frac{\lambda_2}{2}K_G^2 \norm{\Sigma \Sigma^\top},
\end{split}
\end{align}
where the last inequality follows from the generalized PL conditions. Define
\begin{align}
    \mu_1(h) = \min \Big \{\mu(h)^2, \frac{c}{2\lambda_1\eta^2}h \Big\},\, \forall h\ge0 .
\end{align}
Then, $\mu_1 \in \Kcal_\infty$ if $\mu \in \Kcal_\infty$, $\mu_1 \in \Kcal$ if $\mu \in \Kcal$, and $\mu_1 \in \PDcal$ if $\mu \in \PDcal$. 

When $\mu \in \Kcal$, applying the weak triangle inequality (Lemma~\ref{lm:weakTriangle}) yields
\begin{subequations}
\begin{align}
    \Lcal[\Vcal_2](z,v) &\le - \lambda_1 \eta \mu_1 (\Jcal(z) - \Jcal^*) - \lambda_1 \eta \mu_1(\innprod{v}{v}) + \frac{\lambda_2}{2}K_G^2 \norm{\Sigma \Sigma^\top} \label{eq:LV2mu1}\\
    &\le - \lambda_1 \eta \mu_1\Big(\frac{1}{2}(\Jcal(z) - \Jcal^*+\innprod{v}{v})\Big) + \frac{\lambda_2}{2}K_G^2 \norm{\Sigma \Sigma^\top}.
\end{align}    
\end{subequations}
Moreover, since
\begin{align}\label{eq:V2byJ+v}
    \Vcal_2(z,v) \le (1+\lambda_1L)(\Jcal(z) - \Jcal^*) + \frac{\lambda_1 + \lambda_2}{2}\innprod{v}{v}\le \lambda_3 (\Jcal(z) - \Jcal^* + \innprod{v}{v})
\end{align}
with $\lambda_3 = \max\{1+\lambda_1L,  (\lambda_1 + \lambda_2)/2\}$, it follows that
\begin{align}
    \Lcal[\Vcal_2](z,v) \le - \lambda_1 \eta \mu_1\Big(\frac{1}{2\lambda_3}\Vcal_2(z,v)\Big) + \frac{\lambda_2}{2}K_G^2 \norm{\Sigma \Sigma^\top}.
\end{align}
Hence, when $\mu$ is of $\Kcal_\infty$, the function $\Vcal_2$ qualifies as an NSS-Lyapunov function, while if $\mu \in \Kcal$, it qualifies as a scNSS-Lyapunov function. Statements 1 and 2 then follow directly from Theorems~\ref{thm: NSS} and~\ref{thm: scNSS}, respectively.

When $\mu \in \PDcal$, define $\mu_2(r) = \min_{s \in [0,r]}\{\mu_1(s)+\mu_1(r-s) \}$, which is also a positive definite function. Moreover, $\mu_2$ satisfies the weak subadditivity property $\mu_2(r+s) \le \mu_1(r) + \mu_1(s)$. Therefore, combining \eqref{eq:LV2mu1} and \eqref{eq:V2byJ+v}, we obtain
\begin{align}
    \Lcal[\Vcal_2](z,v) \le - \lambda_1 \eta \mu_2\Big(\frac{1}{\lambda_3}\Vcal_2(z,v)\Big) + \frac{\lambda_2}{2}K_G^2 \norm{\Sigma \Sigma^\top}.
\end{align}
It follows that, when $\mu \in \PDcal$, the function $\Vcal_2$ serves as an iISS-Lyapunov function. By Theorem \ref{thm: iNSS}, Statement~3 is established.
\end{proof}

\subsubsection{Robustness without Global Lipchitz Condition}
In the previous subsection, we required a global Lipschitz condition on the gradient to establish the NSS properties of the underdamped Langevin diffusion. However, in many practical scenarios this global Lipschitz condition is difficult to satisfy, which highlights the need to develop NSS guarantees without relying on it. We address this challenge by properly tuning the learning rate $\eta$ and damping coefficient $c$.

To prepare the main results, we aim to establish a global upper bound on the gradient. Since $\Jcal$ is twice continuously differentiable, its gradient is locally Lipschitz and its Hessian is locally bounded. For $h \ge 0$, define 
\begin{align}
\begin{split}
    \bar{L}_2(h) &= \max_{z \in \Zcal_h}\norm{\nabla^2 \Jcal(z)}\\
    \tilde{L}_2(h) &= \bar{L}_2(h) - \bar{L}_2(0),
\end{split}
\end{align}
where $\Zcal_h$ is the sublevel set introduced in \eqref{eq:Zsublevelset}. By construction, $\tilde{L}_2$ is continuous, vanishes at zero and nondecreasing. 

\begin{lemma}\label{lm:gradientUpperBound}
The gradient of $\Jcal$ is upper bounded by
    \begin{align}\label{eq:gradientUpper}
        \norm{\nabla \Jcal(z)}^2 \le \varphi_1(\Jcal(z) - \Jcal^*),
    \end{align}
where $\varphi_1(h) = 2\bar{L}_2(h)h + \frac{5}{2}h$ is a $\Kcal_\infty$-function. If $\varphi_1$ is not continuously differentiable, it can be smoothed by defining
\begin{align}\label{eq:varphiDef}
    \varphi_2(h) = \frac{1}{\delta} \int_{h}^{h+\delta} \Big(2\bar{L}_2(s)s + \frac{5}{2}s \Big) \de s \ge \varphi_1(h), \, \forall h \ge 0.
\end{align}
with any $\delta > 0$. In addition, $\varphi_2$ is a $\Kcal_\infty$-function, continuously differentiable, and satisfies $\varphi_2'(h) \ge 2 \bar{L}_2(h) + 5/2$.
\end{lemma}
\begin{proof}
According to \cite[Lemma 1.2.2]{nesterov2013introductory}, $\Jcal$ is $\bar{L}_2(h)$–smooth on $\Zcal_h$. Now fix $z \in \Zcal_h$ and define 
\begin{align}
    \kappa(z,s) = \Jcal(z - s\nabla \Jcal(z)).
\end{align}
Its derivative at $s=0$ satisfies
\begin{align}
    \frac{\partial \kappa(z,s)}{\partial s}\bigg\lvert_{s=0} = -\innprod{\nabla \Jcal(z)}{\nabla \Jcal(z)} \le 0.
\end{align}
Since $\Zcal_h$ is compact \cite[Lemma 2.4]{Sontag2022}, there exists $\bar{s} > 0$ such that $z - \bar{s}\nabla \Jcal(z)$ first lies on the boundary of $\Zcal_h$, i.e. $\kappa(z,\bar{s}) = h$. 

We will show that $\bar{s} \ge \frac{2}{\bar{L}_2(h)}$ by contradiction. Suppose $\bar{s} < \frac{2}{\bar{L}_2(h)}$. Since $z - {s}\nabla \Jcal(z) \in \Zcal_h$ for all $s \in [0,\bar{s}]$, by \cite[Lemma 1.2.3]{nesterov2013introductory},  we have
\begin{align}
    h = \Jcal(z - \bar{s}\nabla \Jcal(z)) \le \Jcal(z) + \bigg(-\bar{s}  + \frac{\bar{L}_2(h) \bar{s}^2}{2}\bigg)\norm{\nabla \Jcal(z)}^2.
\end{align}
If $\bar{s} < {2}/{\bar{L}_2(h)}$, then $\big(-\bar{s} + {\bar{L}_2(h)\bar{s}^2}/{2}\big) < 0$, so $h < \Jcal(z)$ contradicting $\Jcal(z) \le h$. Hence, $\bar{s} \ge \tfrac{2}{\bar{L}_2(h)}$.

Take $s = \tfrac{1}{\bar{L}_2(h)}$. Since $s \le \bar{s}$, we have $z - s\nabla \Jcal(z) \in \Zcal_h$. Using \cite[Lemma 1.2.3]{nesterov2013introductory} again, 
\begin{align}\label{eq:LipchitzCondition2}
    \Jcal^* \le \Jcal\bigg(z - \frac{1}{\bar{L}_2(h)}\nabla \Jcal(z)\bigg) \le \Jcal(z) - \frac{1}{2\bar{L}_2(h)}\norm{\nabla \Jcal(z)}^2.
\end{align}
Since $z \in \Zcal_{\Jcal(z)-\Jcal^*}$, we can set $h = \Jcal(z) - \Jcal^*$, obtaining 
\begin{align}
    \norm{\nabla \Jcal(z)}^2 \le 2\bar{L}_2(\Jcal(z)-\Jcal^*)(\Jcal(z)-\Jcal^*) \le \varphi_1(\Jcal(z)-\Jcal^*),
\end{align}
where $\varphi_1(h) = 2\bar{L}_2(h)h + \frac{5}{2}h = 2\tilde{L}_2(h)h+2\bar{L}_2(0)h + \frac{5}{2}h$, and it is a $\Kcal_\infty$-function. Hence, inequality \eqref{eq:gradientUpper} holds. 

To ensure continuous differentiability, we smooth $\varphi_1$ by convolution and define $\varphi_2$ as \eqref{eq:varphiDef}. In addition, the derivative of $\varphi_2$ is
\begin{align}
\begin{split}
    \varphi_2'(h) &= \frac{1}{\delta}(\varphi_1(h+\delta) - \varphi_1(h)) = \frac{1}{\delta}(2\bar{L}_2(h+\delta)(h+\delta)  - 2\bar{L}_2(h)h) + \frac{5}{2} \\
    &\ge 2\bar{L}_2(h) +  \frac{5}{2},
\end{split}
\end{align}
where the last inequality follows from the monotonicity of $\bar{L}_2$. Thus, both \eqref{eq:gradientUpper} and the derivative condition are established, completing the proof. 
\end{proof}
The connections between the NSS properties of underdamped Langevin diffusion and the generalized PL conditions are stated below.
\begin{theorem}\label{thm:NSSunderdampedLangevin}
    Let $G(z,v)$ be globally bounded, i.e., $\normF{G(z,v)} \le K_G$ for all $z \in \Zcal$ and $v \in \R^n$, $c = \frac{1}{2} \norm{\nabla^2 \Jcal(z)}+\frac{1}{2}$ and $\eta = \frac{1}{2}(\varphi_2'(\Jcal(z) - \Jcal^*) - c)$. Then:
    \begin{enumerate}
        \item If the function $\Jcal$ satisfies the $\Kcal_\infty$-PL condition, the underdamped Langevin diffusion \eqref{eq:HeavyballFlowStochastic} is NSS;
        \item If the function $\Jcal$ satisfies the $\Kcal$-PL condition, the underdamped Langevin diffusion \eqref{eq:HeavyballFlowStochastic} is scNSS and iNSS.
        \item If the function $\Jcal$ satisfies the $\mathcal{PD}$-PL condition, the underdamped Langevin diffusion \eqref{eq:HeavyballFlowStochastic} is iNSS.
    \end{enumerate}
\end{theorem}
\begin{proof}
The candidate Lyapunov function is designed as 
\begin{align}
    \Vcal_3(z,v) = \varphi_2(\Jcal(z) - \Jcal^*) + \innprod{\nabla \Jcal(z)}{v} + \innprod{v}{v}.
\end{align}
Following Lemma \ref{lm:gradientUpperBound}, and by applying the Cauchy–Schwarz and Young’s inequalities, the Lyapunov function can be bounded by
\begin{align}\label{eq:V3bound}
\begin{split}
    \frac{1}{2}\varphi_2(\Jcal(z) - \Jcal^*) + \frac{1}{2}\innprod{v}{v} \le \Vcal_3(z,v) &\le \frac{3}{2}\varphi_2(\Jcal(z) - \Jcal^*) + \frac{3}{2}\innprod{v}{v}
\end{split}
\end{align}
It is clearly that $\Vcal_3$ is positive definite and coercive on $\Zcal \times \R^n$. 

Now we analyze the $\Lcal$ generator applied to $\Vcal_3$, which is given by
\begin{align}
\begin{split}
    \Lcal[\Vcal_3](z,v) &= \varphi_2'(\Jcal(z) - \Jcal^*)\innprod{\nabla \Jcal(z)}{v} + \innprod{\nabla^2 \Jcal(z) v}{v} \\
    &\quad + \innprod{\nabla \Jcal(z)}{-\eta \nabla \Jcal(z) - cv} + 2\innprod{v}{-\eta \nabla \Jcal(z) - cv} + \innprod{G(z,v)\Sigma}{G(z,v)\Sigma} \\
    &\le -\eta \innprod{\nabla \Jcal(z)}{\nabla \Jcal(z)} -\big(2c - \norm{\nabla^2 \Jcal(z)} \big) \innprod{v}{v} \\
    &\quad + \big( \varphi_2'(\Jcal(z) - \Jcal^*) - c - 2\eta \big)\innprod{\nabla \Jcal(z)}{v}  + K_G^2 \norm{\Sigma \Sigma^\top}.
\end{split}
\end{align}
Following the property of $\varphi_2$ in Lemma \ref{lm:gradientUpperBound}, 
\begin{align}
    \eta = \frac{1}{2}\big(\varphi_2'(\Jcal(z) - \Jcal^*) - \frac{1}{2}\norm{\nabla^2 \Jcal(z)} - \frac{1}{2} \big) \ge 1.
\end{align}
By the chosen $\eta$ and $c$, and the generalized PL conditions, it holds
\begin{align}\label{eq:V3Liederivative}
\begin{split}
     \Lcal[\Vcal_3](z,v) &\le -  \mu(\Jcal(z) - \Jcal^*)^2 -\innprod{v}{v} + K_G^2 \norm{\Sigma \Sigma^\top} \\
     &\le -\mu_3(\varphi_2(\Jcal(z) - \Jcal^*)) - \mu_3(\innprod{v}{v}) + K_G^2 \norm{\Sigma \Sigma^\top},
\end{split}
\end{align}
where $\mu_3$ is defined as
\begin{align}
    \mu_3(h) = \min \{\mu(\varphi_2^{-1}(h))^2, h \} ,\, \forall h\ge0 .
\end{align}
Then, $\mu_3 \in \Kcal_\infty$ if $\mu \in \Kcal_\infty$, $\mu_3 \in \Kcal$ if $\mu \in \Kcal$, and $\mu_3 \in \PDcal$ if $\mu \in \PDcal$.

When $\mu \in \Kcal$, applying the weak triangle inequality (Lemma \ref{lm:weakTriangle}) to \eqref{eq:V3Liederivative} gives
\begin{align}
\begin{split}
    \Lcal[\Vcal_3](z,v) &\le - \mu_3\bigg(\frac{1}{2} \varphi_2( \Jcal(z) - \Jcal^*) + \frac{1}{2}\innprod{v}{v}\bigg) + K_G^2 \norm{\Sigma \Sigma^\top} \\
    &\le - \mu_3\bigg(\frac{1}{3} \Vcal_3(z,v)\bigg) + K_G^2 \norm{\Sigma \Sigma^\top}
\end{split}
\end{align}
where the last inequality is a result of \eqref{eq:V3bound}. When $\Jcal$ satisfies the $\Kcal_\infty$-PL condition, $\mu_3$ is a $\Kcal_\infty$-function, and hence, system \eqref{eq:HeavyballFlowStochastic} is NSS by Theorem \ref{thm: NSS}. When $\Jcal$ satisfies the $\Kcal$-PL condition, $\mu_3$ is a $\Kcal$-function, and hence, system \eqref{eq:HeavyballFlowStochastic} is scNSS by Theorem \ref{thm: scNSS}.

When $\Jcal$ satisfies the $\PDcal$-PL condition, define $\mu_4(h) = \min_{s \in [0,h]}\{\mu_3(s)+\mu_3(h-s) \}$, which is also a positive definite function. Moreover, $\mu_4$ satisfies the weak subadditivity property $\mu_4(h+s) \le \mu_3(h) + \mu_3(s)$. Therefore, we obtain from \eqref{eq:V3Liederivative} that
\begin{align}
    \Lcal[\Vcal_3](z,v) \le - \mu_4\bigg(\frac{2}{3} \Vcal_3(z,v)\bigg) + K_G^2 \norm{\Sigma \Sigma^\top}.
\end{align}
It follows that, when $\mu \in \PDcal$, the function $\Vcal_3$ serves as an iISS-Lyapunov function. By Theorem \ref{thm: iNSS}, Statement~3 is established.

\end{proof}

% \begin{proof}
% ToDo
% % the infinitesimal generator applied to $\Jcal$ satisfies
% % \begin{align}
% % \begin{split}
% %     \Lcal[\Jcal](\xi, \Theta) &\le - \eta(\Jcal(\xi))\mu(\Jcal(\xi))^2 +  \frac{1}{2}\bar{L}(0)\normF{\Theta}^2 + \frac{1}{2}\tilde{L}(\Jcal(\xi)) \normF{\Theta}^2 \\
% %     &\le -\eta(\Jcal(\xi))\mu(\Jcal(\xi))^2 +  \big[\frac{1}{2}\bar{L}(0) + \rho^{-1}\big(\normF{\Theta}^2\big)\big]\normF{\Theta}^2 + \frac{1}{2}\tilde{L}(\Jcal(\xi)) \rho\bigg( \frac{1}{2}\tilde{L}(\Jcal(\xi)) \bigg)
% % \end{split}
% % \end{align}
% % where $\rho \in \Kcal_\infty$ is a design function to be chosen.     
% \end{proof}

\section{Applications}
In this section, we apply the notion of NSS to analyze the robustness of the gradient descent algorithm for two important problem classes: policy optimization for LQR and logistic regression.
\subsection{Policy Optimization for Linear Quadratic Regulator}
We study the policy optimization algorithm for the following linear system:
\begin{align}
    \dot{x}(t) = Ax(t) + Fu(t),\, x(0) = x_0,
\end{align}
where $A \in \R^{n \times n}$ and $F \in \R^{n \times m}$ are constant system matrices; $x(t) \in \R^n$ and $u(t) \in \R^m$ denote the states and control inputs, respectively. The performance index of the system is defined as 
\begin{align}\label{eq:LQRcost}
    \Jcal_1(x_0,u) = \int_{0}^\infty x(t)^\top Q x(t) + u(t)^\top R u(t) \de t
\end{align}
where $Q \in \sym^n_{++}$ and $R \in \sym^n_{++}$. Under the assumption that the pair $(A,F)$ is controllable, the optimal control for \eqref{eq:LQRcost} is given by
\begin{align}
    u^*(t) = - R^{-1}F^\top P^*x(t) = -K^*x(t)
\end{align}
where $P^* \succ 0$ is the solution of the Riccati equation
\begin{align}
    A^\top P^* + P^* A - P^* FR^{-1}F^\top P^* + Q = 0.
\end{align}
Policy optimization seeks the best feedback control $K^*$ over the admissible set 
\begin{align}
    \Gcal = \{K \in \R^{m \times n}|\, (A-FK) \text{ is Hurwitz} \},
\end{align}
by solving
\begin{align}\label{eq:POLQR}
\begin{split}
    \min_{K \in \Gcal} \Jcal_2(K) &= \E_{x_0 \sim \mathcal{N}(0,I_n)} \Jcal_1(x_0, u=-Kx) \\
    &= \Tr{P_K},
\end{split}
\end{align}
where $P_K \in \sym^n_{++}$ is the solution of the Lyapunov equation
\begin{align}
    (A-FK)^\top P_K + P_K (A-FK) + Q + K^\top R K = 0. 
\end{align}

The gradient of the objective function $\Jcal_2$ is given by
\begin{align}
    \nabla \Jcal_2(K) = 2(RK - F^\top P_K)Y_K
\end{align}
where $Y_K \in \sym^n_{++}$ is the solution of the Lyapunov equation
\begin{align}\label{eq:YLyapunov}
    (A-FK)Y_K + Y_K (A-FK)^\top + I_n = 0. 
\end{align}
Let $Y^*$ denote the solution of \eqref{eq:YLyapunov} corresponding to the optimal feedback gain $K^*$, and $\Jcal_2^* = \Jcal_2(K^*)$. The following lemma shows the local Lipchitz continuity of the gradient $\nabla \Jcal_2$.
\begin{lemma}[Lemma 5.3 in \cite{cui2025perturbed}]\label{lm:Lsmoothness}
    The gradient $\nabla \mathcal{J}_2(K)$ is $\bar{L}_3(h)$-Lipschitz continuous over the sublevel set $\mathcal{G}_h = \{K \in \mathcal{G}| \mathcal{J}_2(K)-\Jcal_2^* \le h \}$, with 
    \begin{align}
    \begin{split}
        \bar{L}_3(h) &= \frac{2\norm{R}}{\eigmin{Q}}(\Jcal_2^*+h) + \frac{8a_2\norm{F}\norm{R}}{\eigmin{Q}^2} (\Jcal_2^*+h)^{\frac{5}{2}} + \frac{8\norm{F}(a_1\norm{R}+\norm{F})}{\eigmin{Q}^2} (\Jcal_2^*+h)^3.
    \end{split}
    \end{align}
    where
    \begin{align}
        a_1 = \frac{2\norm{F}}{\eigmin{R}} ,\, a_2 = \Big(\frac{2\norm{A}}{\eigmin{R}}\Big)^{\!{1}/{2}}
    \end{align}
\end{lemma}
The following lemma shows that $\Jcal_2(K) - \Jcal_2^*$ can serve as a size function over $\Gcal$.
\begin{lemma}[Proposition 3.2 and Lemma 3.3 in \cite{bu2020policy}]\label{lm:coercivity}
    The objective function $\Jcal_2(K)$ is coercive and analytic over $\Gcal$.
\end{lemma}

The following lemma shows that the objective function $\Jcal_2$ satisfies the $\Kcal$-PL condition.  
\begin{lemma}[$\Kcal$-PL condition of LQR \cite{CJS2024}]\label{lm:KPL_LQR}
    The objective function $\mathcal{J}_2(K)$ satisfies the $\mathcal{K}$-PL condition, that is 
\begin{align*}
    \normF{\nabla \mathcal{J}_2(K)} \geq  \mu_5(\mathcal{J}_2(K) - \mathcal{J}_2^*), \quad \forall K \in \mathcal{G},
\end{align*}
where 
\begin{align}
      \mu_5(h) = \frac{h}{b_1h + b_2}, \, \forall h \ge 0
\end{align}
\begin{align*}
b_1 = \frac{\norm{F}\sqrt{2(\eigmin{Y^*}+\eigmax{Y^*})}}{\eigmin{R} \sqrt{\eigmin{Y^*}}}, 
\end{align*}
and 
\begin{align*}
&b_2 = \frac{\norm{A-FK^*}_F^2\eigmin{Y^*}^{{1}/{2}}(\eigmin{Y^*}+\eigmax{Y^*})^{{1}/{2}}}{\sqrt{2}\norm{F}}
\end{align*}
\end{lemma}
The overdamped Langevin diffusion for solving the policy optimization problem \eqref{eq:POLQR} is given by
\begin{align}\label{eq:LQRdiffusion}
    \de K(s) = - 2\eta(\Jcal(K(s)) - \Jcal^*)(RK(s) - F^\top P(s)) Y(s) \de s + \Sigma_1(s) \de W(s)
\end{align}
where $W(s) \in \R^{m \times n}$ is a standard, independent Brownian motion, $P(s)=P_{K(s)}$, $Y(s)=Y_{K(s)}$, and $\Sigma_1: \R_+ \to \R^{m \times m}$ is Borel measurable and locally essentially bounded, which characterizes the intensity of the stochastic perturbations. Since the gradient of $\Jcal_2$ is not globally Lipschitz continuous, the NSS properties of \eqref{eq:LQRdiffusion} can be established by Theorem~\ref{thm:NSSwithoutLipchitz}.

\begin{theorem}
    The overdamped Langevin diffusion in \eqref{eq:LQRdiffusion} is NSS if $\lim_{h \to \infty}\frac{h^3}{\eta(h)} = 0$, and it is scNSS if $0 < \lim_{h \to \infty}\frac{h^3}{\eta(h)} < \infty$. 
\end{theorem}
\begin{proof}
    By Lemma \ref{lm:coercivity}, $\Jcal_2(K) - \Jcal_2^*$ qualifies as a size function. Since $\lim_{h \to \infty}\frac{h^3}{\eta(h)} = 0$ implies that $\lim_{h \to \infty}  \frac{\bar{L}_3(h)}{\eta(h)\mu_5(h)^2}  = 0$, the stochastic dynamics of \eqref{eq:LQRdiffusion} is NSS by Statement~1 of Theorem~\ref{thm:NSSwithoutLipchitz}.

    When $0<\lim_{h \to \infty}\frac{h^3}{\eta(h)} < \infty$, we have $0<\lim_{h \to \infty}  \frac{\bar{L}_3(h)}{\eta(h)\mu_5(h)^2} < \infty$, too. Hence, the stochastic dynamics of \eqref{eq:LQRdiffusion} is scNSS by Statement~2 of Theorem~\ref{thm:NSSwithoutLipchitz}.
\end{proof}

By introducing the momentum in optimization, the heavy-ball method of policy optimization for LQR can be represented as
\begin{align}\label{eq:LQRheavyballdiffusion}
\begin{split}
    \de K(s) &= \de V(s) \\
    \de V(s) &=- 2\eta(RK(s) - F^\top P(s)) Y(s) \de s - cV(s) \de s + \Sigma_1(s) \de W(s).
\end{split}
\end{align}
Its scNSS property can be guaranteed by the following theorem.
\begin{theorem}
    The underdamped Langevin diffusion in \eqref{eq:LQRheavyballdiffusion} is scNSS if $c = \frac{1}{2} \norm{\nabla^2 \Jcal_2(K(s))}+\frac{1}{2}$ and $\eta = \frac{1}{2}(\varphi_3'(\Jcal_2(z) - \Jcal_2^*) - c)$, where  $\varphi_3(h) = 2\bar{L}_3(h)h + 5/2 h$.
\end{theorem}
\begin{proof}
    Since $\Jcal_2$ satisfies the $\Kcal$-PL condition (Lemma \ref{lm:KPL_LQR}), it follows from Statement~2 of Theorem~\ref{thm:NSSunderdampedLangevin} that \eqref{eq:LQRheavyballdiffusion} is scNSS.
\end{proof}

\subsection{Logistic Regression}
Suppose we are given a dataset of $N$ samples $\{(x_i,y_i)\}_{i=1}^N$, where each feature vector $x_i \in \R^n$ and the corresponding label $y_i \in \{0,1\}$. In logistic regression, the probability that a sample with feature vector $x$ belongs to class $y=1$ is modeled as
\begin{align}
    p(x;\theta) = \mathbb{P}(y=1|x;\theta) = \sigma(\theta^\top x) = \frac{1}{1+\exp(-\theta^\top x)},
\end{align}
where $\theta \in \R^n$ denotes the parameter vector to be learned. The model parameters are obtained by minimizing the empirical negative log-likelihood (binary cross-entropy) loss:
\begin{align}\label{eq:logisticLoss}
    \min_{\theta \in \R^n}\Jcal_3(\theta) = \frac{1}{N}\sum_{i=1}^N \ell_i(\theta) = -\frac{1}{N} \sum_{i=1}^N y_i \log(p_i(\theta)) + (1-y_i) \log(1-p_i(\theta)),
\end{align}
where $\ell_i(\theta)$ is the per-sample loss and $p_i(\theta) = p(x_i;\theta)$. Let $X = [x_1,x_2,\cdots,x_n] \in \R^{n \times N}$ denote the data matrix. The gradient and Hessian of the logistic loss are given by
\begin{subequations}
\begin{align}
    \nabla \Jcal_3(\theta) &= \frac{1}{N} \sum_{i=1}^N (p_i(\theta) - y_i)x_i, \label{eq:gradientJ3} \\
    \nabla^2 \Jcal_3(\theta) &= \frac{1}{N} \sum_{i=1}^N p_i(\theta) (1- p_i(\theta))x_ix_i^\top = \frac{1}{N}X \Lambda(\theta) X^\top, \label{eq:HessianJ3}
\end{align}    
\end{subequations}
where $\Lambda(\theta) = \operatorname{diag}(p_i(\theta)(1-p_i(\theta)))$. The following standard assumptions are imposed throughout the analysis.
\begin{assumption}\label{ass:fullrank}
    The dataset is sufficiently rich so that the data matrix $X \in \mathbb{R}^{n\times N}$ has full row rank, i.e., $\operatorname{rank}(X)=n$.
\end{assumption}

\begin{assumption}\label{ass:nonseparable}
    The dataset is \emph{nonseparable}. That is, there does not exist a nonzero vector $\vec{\theta} \in \R^n$, such that $\vec{\theta}^\top x_i \ge 0$ for all $i$ with $y_i = 1$ and $\vec{\theta}^\top x_i \le 0$ for all $i$ with $y_i = 0$. 
\end{assumption}
\subsubsection{Properties of Logistic Loss}
The following lemmas establish key properties of $\Jcal_3$. In particular, under Assumption \ref{ass:nonseparable}, the objective function is shown to be coercive.
\begin{lemma}\label{lm:coercivityLogistic}
    Under Assumption \ref{ass:nonseparable}, the objective function $\Jcal_3$ is coercive over $\R^n$.
\end{lemma}
\begin{proof}
For any $\theta_0 \in \R^n$ and $\theta \neq \theta_0$, write
\begin{align}\label{eq:thetaDecom}
\theta = \theta_0 + r \vec{\theta},\, \vec{\theta} = \frac{\theta - \theta_0}{\norm{\theta - \theta_0}}, \,r = \norm{\theta - \theta_0} .
\end{align}   
Define, for $r\ge 0$ and any unit vector $\vec{\theta}$,
\begin{align}\label{eq:xiDef}
    \zeta(r,\vec{\theta}) = \frac{1}{N}\sum_{i=1}^N\big( p_i(\theta_0 + r \vec{\theta}) - y_i\big) \vec{\theta}^\top x_i,
\end{align}
so that
\begin{align}\label{eq:derivativeXi}
    \frac{\partial}{\partial r}\zeta(r,\vec{\theta}) = \frac{1}{N}\sum_{i=1}^N\big( p_i(\theta_0 + r \vec{\theta})(1-p_i(\theta_0 + r \vec{\theta}))\big) (\vec{\theta}^\top x_i)^2.
\end{align}
When $r \to \infty$:
\begin{align}\label{eq:xiInfinity}
\begin{split}
    \lim_{r \to \infty} \zeta(r,\vec{\theta}) &= \frac{1}{N}\sum_{i=1}^N\big( \mathbf{1}\{\vec{\theta}^\top x_i > 0 \}- y_i \big)\vec{\theta}^\top x_i \\
    &=\frac{1}{N} \sum_{i=1}^N \mathbf{1}\{y_i = 1 \} [-\vec{\theta}^\top x_i ]_+ + \mathbf{1}\{y_i = 0 \} [\vec{\theta}^\top x_i]_+ =: \bar{\zeta}(\vec{\theta}),
\end{split}
\end{align}
with $[s]_+ = \max\{s,0\}$. By Assumption \ref{ass:nonseparable}, for every unit vector $\vec{\theta}$, there is at least one index $i$ with either $y_i=1$ and $\vec{\theta}^\top x_i < 0$ or $y_i=0$ and $\vec{\theta}^\top x_i> 0$. Hence, $\bar{\zeta}(\vec{\theta}) > 0$ for all unit vector $\vec{\theta}$. In addition, it follows from \eqref{eq:derivativeXi} that $\frac{\partial}{\partial r}\zeta(r,\vec{\theta})>0$ and $\zeta(r,\vec{\theta})$ is strictly increasing in $r$ for any fixed $\vec{\theta}$. The map $\bar{\zeta}$ is continuous on the compact unit sphere, so
\begin{align}
    e_1 =  \min_{\norm{\vec{\theta}}=1} \bar{\zeta}(\vec{\theta}) > 0.
\end{align}
Since $\zeta(r,\vec{\theta})$ is continuous in $\vec{\theta}$, strictly increasing in $r$, and converges pointwise to the continuous limit $\bar{\zeta}(\vec{\theta})$, Dini’s theorem yields uniform convergence on the unit sphere. Thus, there exists $0 < r_1 < \infty$, such that
\begin{align}
     \zeta(r,\vec{\theta})\ge \frac{e_1}{2}, \, \forall r \ge r_1, \,  \norm{\vec{\theta}}=1.
\end{align}

Along any ray $\theta_0 + r \vec{\theta}$, 
\begin{align}
    \frac{\partial}{\partial r}\big(\Jcal_3(\theta_0 + r \vec{\theta}) - \Jcal_3(\theta_0 + r_1 \vec{\theta})\big) = \zeta(r,\vec{\theta}), 
\end{align}
and integrating from $r_1$ to $r \ge r_1$ gives
\begin{align}
    \Jcal_3(\theta_0 + r \vec{\theta}) - \Jcal_3(\theta_0 + r_1 \vec{\theta}) = \int_{r_1}^r \zeta(s,\vec{\theta}) \de s \ge \frac{e_1}{2}(r-r_1),\, \forall r\ge r_1.
\end{align}
Because $\vec{\theta} \rightarrow \Jcal_3(\theta_0 + r_1 \vec{\theta})$ is continuous on the unit sphere, 
\[
\min_{\norm{\vec{\theta}}=1} \Jcal_3(\theta_0 + r_1 \vec{\theta}) = e_2 > -\infty.
\] 
Therefore, $ \lim_{r\to\infty} \Jcal_3(\theta_0 + r \vec{\theta}) = \infty$ uniformly in direction $\vec{\theta}$, i.e., the logistic loss is coercive.
\end{proof}

We now present a lemma that demonstrates the convexity and smoothness properties of $\Jcal_3$.
\begin{lemma}\label{lm:LipchitzLogistic}
    Under Assumption~\ref{ass:fullrank}, the objective function $\Jcal_3$ is strictly convex, and its gradient is globally $\tfrac{1}{4N}\,\|XX^\top\|$-Lipschitz continuous. 
\end{lemma}
\begin{proof}
    For a single sample with logit $s_i=\theta^\top x_i$, the second derivative of the loss is
    \begin{align}
        \frac{\mathrm{d}^2 \ell_i(s_i)}{\mathrm{d} s_i^2} = \sigma(s_i)\bigl(1-\sigma(s_i)\bigr) > 0,
    \end{align}
    where $\sigma(\cdot)$ is the logistic sigmoid. Hence, $\ell_i$ is a convex function of $s_i$. Since $s_i$ is affine in $\theta$, each per-sample loss is convex in $\theta$. The overall loss $\Jcal_3$
    is therefore convex.

    The Hessian of $\Jcal_3$ is given in \eqref{eq:HessianJ3}. Each diagonal entry of $\Lambda(\theta)$ satisfies $0< p_i(\theta)(1-p_i(\theta)) \le \tfrac{1}{4}$. Under Assumption~\ref{ass:fullrank}, for any nonzero $v\in\mathbb{R}^n$ we have
    \begin{align}
        v^\top \nabla^2 \Jcal_3(\theta) v
        &= \tfrac{1}{N} \| \Lambda(\theta)^{1/2} X^\top v \|^2 > 0,
    \end{align}
    which implies $\nabla^2 \Jcal_3(\theta) \succ 0$. Thus, $\Jcal_3$ is strictly convex.

    Finally, to show Lipschitz continuity of the gradient, observe
    \begin{align}
        \|\nabla^2 \Jcal_3(\theta)\|
        = \Bigl\|\tfrac{1}{N} X \Lambda(\theta) X^\top \Bigr\|
        \le \tfrac{1}{N} \|\Lambda(\theta)\| \, \|XX^\top\|
        \le \tfrac{1}{4N}\,\|XX^\top\|,
    \end{align}
    since $\|\Lambda(\theta)\|\le \tfrac{1}{4}$. Therefore, the gradient $\nabla \Jcal_3$ is globally $\tfrac{1}{4N}\,\|XX^\top\|$-Lipschitz continuous.
\end{proof}

Since $\Jcal_3$ is coercive over $\R^n$ and strictly convex, there exists a unique minimizer $\theta^*$. The following lemma lays the foundation for establishing the scNSS property of the gradient-based optimization algorithms used to solve logistic regression.
\begin{lemma}\label{eq:logisticPL}
Under Assumptions \ref{ass:nonseparable}, the objective function $\Jcal_3$ satisfies the $\Kcal$-PL condition, i.e.,
\begin{align}\label{eq:LogisticKPL}
    \|\nabla \Jcal_3(\theta)\| \;\ge\; \mu_6 \,\big(\Jcal_3(\theta) - \Jcal_3^*\big),
\end{align}
where $\mu_6$ is a $\Kcal$-function.
\end{lemma}
\begin{proof}
Set $\theta_0=\theta^*$, where $\theta^*$ is a (finite) stationary point of $\Jcal_3$. Define $\zeta^*$ as in \eqref{eq:xiDef} with $\theta_0$ replaced by $\theta^*$.
From \eqref{eq:gradientJ3}, we have 
\begin{align}\label{eq:lossgradient1}
    \norm{\nabla \Jcal_3(\theta)}^2 = \frac{1}{N^2} \Big\|\sum_{i=1}^N\big( p_i(\theta^* + r \vec{\theta}) - y_i\big) x_i\Big\|^2 \ge  \frac{1}{N^2} \Big(\sum_{i=1}^N\big( p_i(\theta^* + r \vec{\theta}) - y_i\big) \vec{\theta}^\top x_i \Big)^2 = \zeta^*(r,\vec{\theta})^2.
\end{align}
Since $\theta^*$ is a finite stationary point of $\Jcal_3$, $\zeta(0,\vec{\theta}) = 0$ for all unit vector $\vec{\theta}$. Moreover, by \eqref{eq:derivativeXi} and \eqref{eq:xiInfinity}, for each fixed direction $\vec{\theta}$, $r \rightarrow \zeta^*(r,\vec{\theta})$ is strictly increasing and saturates when $r \to \infty$. Hence, $\zeta^*(\cdot ,\vec{\theta})$ is a $\Kcal$-function with saturation.

Next, note that along any ray,
\begin{align}
    \frac{\partial}{\partial r}\big(\Jcal_3(\theta^* + r \vec{\theta}) - \Jcal_3(\theta^*)\big) = \zeta^*(r,\vec{\theta}). 
\end{align}
By integration, we have
\begin{align}
    \Jcal_3(\theta^* + r \vec{\theta}) - \Jcal_3(\theta^*) = \int_{0}^r \zeta^*(s,\vec{\theta}) \de s =: \psi(r,\vec{\theta}).
\end{align}
Because $\zeta^*(\cdot, \vec{\theta})$ is increasing and has a positive limit, $\psi(\cdot, \vec{\theta})$ is a $\Kcal_\infty$-function for each fixed $\vec{\theta}$. Therefore, for each fixed $\vec{\theta}$, the inverse of $\psi(\cdot, \vec{\theta})$, denoted as $\psi^{-1}(\cdot, \vec{\theta})$ exists and belongs to class $\Kcal_\infty$. By \eqref{eq:lossgradient1}, it holds
\begin{align}\label{eq:lossgradient2}
    \norm{\nabla \Jcal_3(\theta)} \ge \zeta^*(r,\vec{\theta}) = \zeta^*\Big(\psi^{-1}\big(\Jcal_3(\theta^* + r \vec{\theta}) - \Jcal_3(\theta^*), \vec{\theta}\big), \vec{\theta} \Big).
\end{align}

Finally, define a direction-uniform lower envelop 
\begin{align}
    \mu_6(r) = \min_{\norm{\vec{\theta}}=1}\zeta^*(\psi^{-1}(r,\vec{\theta}), \vec{\theta}). 
\end{align}
Because $\zeta^*(\psi^{-1}(r,\vec{\theta}),\vec{\theta})$ is smooth and strictly increasing in $r$ for each fixed $\vec{\theta}$, and the unit sphere of $\vec{\theta}$ is compact, $\mu_6$ is continuous, strictly increasing, and satisfies $\mu_6(0)=0$. Hence, $\mu_6$ is a $\Kcal$-function. Substituting it into the bound \eqref{eq:lossgradient2} yields \eqref{eq:LogisticKPL}, which establish the $\Kcal$-PL condition of the logistic loss.

\end{proof}
\begin{remark}
    Even though the logistic loss $\Jcal_3$ is strictly convex (Lemma \ref{lm:LipchitzLogistic}), it cannot satisfy the $\Kcal_\infty$-PL condition. Indeed, since $|p_i(\theta) - y_i| \le 1$, we have 
    \[
    \|\nabla \Jcal_3(\theta)\| \;\le\; {\|X\|}/{\sqrt{N}}.
    \]
    On the other hand, by Lemma \ref{lm:coercivityLogistic}, $\Jcal_3(\theta) \to \infty$ as $\|\theta\| \to \infty$. Hence, it is impossible to find a $\Kcal_\infty$-function $\mu$ such that
    \begin{align}
        \|\nabla \Jcal_3(\theta)\| \;\ge\; \mu(\Jcal_3(\theta) - \Jcal_3^*).
    \end{align}
    In addition, since $\lim_{\norm{\theta} \to \infty}p_i(\theta)(1-p_i(\theta)) = 0$, it follows from \eqref{eq:HessianJ3} that $\lim_{\norm{\theta} \to \infty} \nabla^2 \Jcal_3(\theta) = 0$. Hence, the logistic loss is not globally strongly convex either. 
\end{remark}

\subsubsection{Robustness Analysis}
Under stochastic perturbations, the first-order gradient flow for solving logistic regression \eqref{eq:logisticLoss} can be represented by the following overdamped Langevin diffusion:
\begin{align}\label{eq:logisticLangevin}
    \de \theta(s) \;=\; - \nabla \Jcal_3\big(\theta(s)\big)\, \de s \;+\; \Sigma(s)\, \de B(s),
\end{align}
where $B(s)$ denotes a standard Brownian motion and $\Sigma(s)$ characterizes the covariance structure of the noise. As a direct application of the results developed in Section \ref{sec:NSSStochasticGradient}, it follows that the Langevin diffusion in \eqref{eq:logisticLangevin} is scNSS.
\begin{theorem}
    Under Assumptions \ref{ass:fullrank} and \ref{ass:nonseparable}, the overdamped Langevin diffusion in \eqref{eq:logisticLangevin} is scNSS.
\end{theorem}
\begin{proof}
    By Lemma \ref{lm:coercivityLogistic}, the objective function $\Jcal_3(\theta) - \Jcal_3^*$ qualifies as a size function. Lemma \ref{lm:LipchitzLogistic} ensures that the gradient of $\Jcal_3$ is globally Lipschitz continuous. Moreover, by Lemma \ref{eq:logisticPL}, $\Jcal_3$ satisfies the $\Kcal$-PL condition. Hence, applying Theorem \ref{thm:overLangevinLipchitz}, the stochastic dynamics in \eqref{eq:logisticLangevin} is scNSS.
\end{proof}

If the heavy-ball gradient descent algorithm is applied to logistic regression under stochastic noise with time-varying covariance, the resulting dynamics can be represented by the following underdamped Langevin system:
\begin{align}\label{eq:logisticDampLangevin}
\begin{split}
    \de \theta(s) &= \de v(s) \\
    \de v(s) &= - \eta \nabla \Jcal_3(\theta(s)) \de s -cv(s)\de s  + \Sigma(s) \de B(s)
\end{split}
\end{align}
where $v(s)$ denotes the velocity, $\eta > 0$ is the learning rate, $c > 0$ is the damping coefficient, $\Sigma(s)$ encodes the time-varying noise covariance, and $B(s)$ is a standard Brownian motion. Using the main results of Theorem \ref{thm:DampedLangevinLipchitz}, we conclude that the Langevin system \eqref{eq:logisticDampLangevin} is scNSS.
\begin{theorem}
    Under Assumptions \ref{ass:fullrank} and \ref{ass:nonseparable}, the underdamped Langevin diffusion in \eqref{eq:logisticDampLangevin} is scNSS.
\end{theorem}
\begin{proof}
    By Lemmas \ref{lm:coercivityLogistic} and \ref{lm:LipchitzLogistic}, the objective function $\Jcal_3$ is coercive on $\R^n$ and has a globally Lipchitz-continuous gradient. In addition, $\Jcal_3$ satisfies the $\Kcal$-PL condition (Lemma \ref{eq:logisticPL}). Therefore, by Theorem \ref{thm:DampedLangevinLipchitz}, the underdamped Langevin diffusion in \eqref{eq:logisticDampLangevin} is scNSS.
\end{proof}

\section{Conclusions}
In this paper, we introduced a new notion—small-covariance NSS—alongside a Lyapunov condition for it. Small-covariance NSS ensures that if the covariance of the stochastic noise is sufficiently small, the trajectories of a stochastic system will eventually enter and remain inside a neighborhood of the equilibrium. The size of this neighborhood is proportional (in a generally nonlinear way) to the magnitude of the covariance. Under the NSS framework, we studied stochastic gradient dynamics (including overdamped and underdamped Langevin diffusion) and showed that if the objective function satisfies the generalized $\Kcal$-PL condition, the stochastic gradient dynamics is small-covariance NSS. In addition, if the $\Kcal$-PL condition is strengthened to a $\Kcal_\infty$-PL condition, the stochastic gradient dynamics is NSS; if the $\Kcal$-PL condition is weakened to a positive-definite condition, the stochastic gradient dynamics is integral NSS. The developed theoretical framework was applied to policy optimization for LQR and to logistic regression, and it is shown that the stochastic gradient dynamics for both are small-covariance NSS. Future research will focus on robustness analysis of Nesterov’s accelerated gradient method and on methodologies to enhance the robustness of gradient-based methods.

\bibliographystyle{alpha}
\bibliography{reference}

\end{document}